\documentclass[ssy]{imsart}
\usepackage{amsfonts}

\RequirePackage[OT1]{fontenc}
\RequirePackage{amsthm,amsmath}
\RequirePackage[numbers]{natbib}
\RequirePackage[psamsfonts]{amssymb}
\RequirePackage{graphicx}
\RequirePackage{tikz}
\RequirePackage{subfigure}
\usetikzlibrary{arrows,automata}
\arxiv{arXiv:0000.0000}
\startlocaldefs
\numberwithin{equation}{section}
\theoremstyle{plain}
\newtheorem{theorem}{Theorem}[section]
\newtheorem{definition}[theorem]{Definition}
\newtheorem{claim}[theorem]{Claim}
\newtheorem{lemma}[theorem]{Lemma}
\newtheorem{proposition}[theorem]{Proposition}
\newtheorem{remark}[theorem]{Remark}
\newtheorem{corollary}[theorem]{Corollary}
\newcommand{\sgn}{\operatorname{sgn}}
\endlocaldefs

\begin{document}

\begin{frontmatter}

\title{Cascade Markov Decision Processes: Theory and Applications}
\runtitle{Cascade MDPs}

\begin{aug}
  \author{\fnms{Manish}  \snm{Gupta}\corref{}\thanksref{t2}\ead[label=e1]{mgupta@fas.harvard.edu}},
  
  \thankstext{t2}{Ph.D Candidate in Applied Mathematics, Harvard School of Engineering and Applied Sciences.}

  \runauthor{M. Gupta.}

  \affiliation{Harvard School of Engineering and Applied Sciences}

  \address{31 Oxford Street, Cambridge MA 02138,\\ 
          \printead{e1}}

\end{aug}

\begin{abstract}
This paper considers the optimal control of time varying continuous time
Markov chains whose transition rates are themselves Markov processes. In one
set of problems the solution of an ordinary differential equation is shown to determine the
optimal performance and feedback controls, while some other cases are
shown to lead to singular optimal control problems which are more difficult to solve. Solution techniques are
demonstrated using examples from finance to behavioral decision making.
\end{abstract}

\begin{abstract}
\end{abstract}



\end{frontmatter}

\section{\label{SecIntro}Introduction}

For over five decades the subject of control of Markov processes has enjoyed
tremendous successes in areas as diverse as manufacturing, communications,
machine learning, population biology, management sciences, clinical systems
modelling and even human memory modeling \cite{OPTHASSL}. In a Markov
decision process (MDP) the transition rates depend upon controls, which can
be chosen appropriately so as to achieve a particular optimization goal. The
subject of this paper is to explore a class of MDPs where the transition
rates are, in addition, dependent upon the state of another stochastic
processes and are thus Markov processes themselves. Our purpose is to
describe a broad range of optimal control problems in which these so-called 
\textit{cascade Markov decision processes} (CMDP)\ admit explicit solutions 
\cite{BROCKOPT}, as well as problems in which dynamic programming is not
applicable at all.

Cascade processes are ideal in modeling games against nature. An epidemic
control system where infection rates vary in accordance with uncontrollable
factors such as the weather is one such case. They are also applicable in
behavioral models of decision making where available choices at each step
may be uncertain. For example, a behavioral decision-making problem called
the "Cat's Dilemma" first appeared in \cite{Makowski:896829} as an attempt
to explain "irrational" choice behavior in humans and animals where observed
preferences seemingly violate the fundamental assumption of transitivity of
utility functions \cite{KALENSCH},\cite{ISI:000229309600009}. In this
problem, each day the cat needs to choose one among the many types of food
presented to it so as to achieve a long-term balanced diet goal. However,
the pet owner's daily selection of food combinations presented to the cat is
random. The cat's feeding choice forms a controlled Markov chain, but the
available foods themselves are contingent on the owner's whim. Another
example is found in \textit{dynamically-hedged} portfolio optimization,
where dynamic (stochastic) rebalancing of allocated weights can be modeled
as a controlled Markov chain. However, what reallocations are possible may
depend on the current prices of assets, which are themselves stochastic.
Such MDP models have the advantage, for example, of being more realistic
than their \textit{continuously-hedged} counterpart, which have
traditionally been studied using Gauss/Markov models on augmented state
spaces \cite{5429700}, \cite{Noh2011510}. Other examples where CMDP are
applicable include queuing systems where service times depend on the 
state of another queue and models of resource sharing where one process
requires exclusivity and another doesn't (e.g., determining the optimal sync
rate for an operating system).

While a cascade Markov process can be equivalently represented on the joint
(coupled) state space as a non-cascade, the main purpose of this paper is to
investigate solutions on \emph{decomposed }state spaces. The main
contributions in doing so include:

\begin{itemize}
\item Decoupled matrix differential equations as solutions to a variety of
fully observable cascade problems involving optimization of the expectation
of a utility functional, which are computationally easier to implement than
their non-decoupled counterpart, and require solving of a one-point instead
of a two-point boundary value problem. 

\item Reduction of a partially observable cascade optimal control problem to
a lower dimensional non-cascade problem (via a process we call \emph{%
diagonalization }) that facilitates the use of standard optimization
techniques on a reduced state space, thereby circumventing the "curse of
dimensionality".  

\item Simpler analysis, via diagonalization, of a class of problems those
that involve optimization of a non-linear function of expectation (such as a
fairness or diversity index) and a full solution to a particular example of
such \emph{singular }optimal control problems.

\item A simple toy model for the dynamically-hedged portfolio optimization
problem and solutions that can be easily generalized to computationally
feasible algorithms for optimal allocation of large scale portfolios.
\end{itemize}

In addition to having the advantages of being able to efficiently represent
large state space Markov processes by factorization to simpler lower
dimensional problems and thus derive computationally simpler solutions, our
approach of working decomposed representations is generalizable to
multi-factor processes, stochastic automata networks \cite{Plateau97stochasticautomata}, 
and even quantum Markov chains and controls 
\cite{ISI000258175300005},\cite{BELLAVK}. 

The particular framework of Markov decision processes closely follows the
assumptions and modeling of \cite{BROCKOPT}, which are characterized by
finite or denumerably many states with perfect state observations and affine
dependence of transition rates on controls. The paper is organized as
follows. A mathematical framework is first outlined, more details of which
are in Appendix \ref{SectionProductMP}. We then derive solutions to two classes of optimal
control problems. In the first case the cost function is a the expectation
of a functional, one that can be solved by dynamic programming requiring
solution to a one point boundary value problem. The second class is the case
where the cost function can not be written as an expectation, a rather
non-standard stochastic control problem but one that arises in applications
requiring diversification (entropy) maximization or variance minimization
and requires solution to two-point boundary value problems. In many cases
the latter is a singular optimal control problem. We will then discuss toy
examples in each class of problems: a portfolio optimization problem and
animal behavior (decision-making) problem. More examples of portfolio
optimization and their cascade solutions appear in the Appendix. 

\section{Cascade Markov Decision Processes}

\subsection{\label{SecMarkovProcessModel}Markov Decision Process Model}

We use the framework of \cite{BROCKOPT} for continuous-time finite-state
(FSCT) Markov processes. We assume a probability space $(\Omega ,\mathcal{F},%
\mathbb{P})$ and right-continuous stochastic processes adapted to a
filtration $\mathbb{F=(}\mathcal{F}_{t})_{t\in T}$ on this space. An FSCT
Markov process $x_{t}$ that is assumed to take values in $\{e_{i}\}_{i=1}^{n}
$, the set of $n$standard basis vectors in $\mathbb{R}^{n}$, has the
following sample path (It\^o) description:
descriptions: 
\begin{eqnarray}
dx &=&\sum_{i=1}^{m}G_{i}xdN_{i}~~  \label{ItoMC} \\
\end{eqnarray}
where $G_{i}\in \mathbb{G}^{n}$ are \textit{distinct}\footnote{%
If the $G_{i}^{\prime }s$ are not distinct, then one can combine the Poisson
counters corresponding to identical $G_{i}^{\prime }s$ to get a set of
distinct $G_{i}^{\prime }s$. For example, $G_{1}ydN_{1}+G_{1}ydN_{2}$ can be
replaced by $G_{1}ydN$ where $dN=dN_{1}+dN_{2},$ a Poisson counter with rate
equal sum of the rates of the counters $N_{1},N_{2}$}, $\mathbb{G}^{n}$ %
being the space of square $n-$matrices of the form $F_{kl}-F_{ll}$ ~where $%
F_{ij}$ is the matrix of all zeros except for one in the $i^{\prime }th$ row
and $j^{\prime }th$ column, and $N_{i}$ are Poisson counters with rates $%
\lambda _{i}.$ The resulting \emph{infinitesimal generator~}that governs the
transition probabilities of the process is $P\in \mathbb{P}^{n}$, the space
of all stochastic $n-$matrices and is given by: 
\begin{eqnarray*}
P &=&\sum_{i=1}^{m}G_{i}\lambda _{i}
\end{eqnarray*}%
In a Markov decision process, the transition
rates are allowed to depend on $\mathcal{F}_{t}-$progressively measurable
control processes $u=(u_{1,}u_{2}...u_{p})$ in an affine accordance with\footnote{that is, we assume an affine dependence on controls}:
\begin{equation*}
\lambda_{i}= \lambda _{i0}+\sum_{j=1}^{p}\mu_{ij}u_{j}
\end{equation*}
so that the infinitesimal generator can be written as:
\begin{equation*}
P(u)=\sum_{i=1}^{m}G_{i}\left( \lambda _{i0}+\sum_{j=1}^{p}\mu
_{ij}u_{j}\right) 
\end{equation*}

\subsection{\label{SectionCascadeMDPModel1}Cascade MDP Model}

We are interested in the case where transition rates of $x_{t}\in
\{e_{i}\}_{i=1}^{n}$ are themselves stochastic: specifically, they depend on
the state of another Markov process, say, $z_{t}\in \{e_{i}\}_{i=1}^{r}$. We
will call such a pair to form a \textbf{Cascade Markov chain~(}CMC)\textbf{\ 
}In general, various levels of interactions between two processes $x_{t}$
and $z_{t}$ defines a joint Markov process $y_{t}=z_{t}\otimes x_{t}$ that
evolves on the product space $\{e_{i}\}_{i=1}^{n}\times \{e_{i}\}_{i=1}^{r}$
(see Appendix A) but we are specifically interested in CMCs where sample
paths of $z_{t}$ and $x_{t}$ have the following have the following Ito
description (Proposition \ref{ProposDecomposIto}, Appendix A): 
\begin{eqnarray}
dz &=&\sum_{i=1}^{s}H_{i}zdM_{i}  \label{CascadeDecRep} \\
dx(z) &=&\sum_{i=1}^{m}G_{i}(z)xdN_{i}(z)~
\end{eqnarray}%
where $H_{i}\in \mathbb{G}^{r},~G_{i}(z)\in \mathbb{G}^{n}$ and the rates of
Poisson counters $M_{i}$ and $N_{i}$ are $\nu _{i}$ and $\lambda _{i}$ with $%
\lambda _{i}$ depending on the state of $z_{t}$. Thus the infinitesimal
generators $P$ and $C$ of $x_{t}$ and $z_{t}$ ($P$ depends on $z_{t}$ and $%
P(z)$ propagates the \emph{conditional} probabilities of $x_{t}$ given $z$)
are 
\begin{eqnarray}
P(z) &=&\sum_{i=1}^{m}G_{i}\lambda _{i}(z)\ ~  \label{CascadeTranProbRep} \\
~~C &=&\sum_{i=1}^{s}H_{i}\nu_{i} 
\end{eqnarray}%
In a \textbf{Cascade Markov decision process~(CMDP)}, we assume the rates $%
\lambda _{i}$ of counters $N_{i}$  are allowed to additionally depend on $%
\mathcal{F}_{t}-$progressively measurable control processes $%
u=(u_{1,}u_{2}...u_{p})$ in accordance with \footnote{%
Each term is, in additional, a function of time $t$ but for clarity explicit
dependence on $t$ will not be specified in notation.}%
\begin{equation*}
\lambda _{i}(z)=\lambda _{i0}^{0}+\lambda _{i0}(z)+\sum_{j=1}^{p}\mu
_{ij}(z)u_{j}
\end{equation*}%
so that the conditional probability vector $p(z,u)$ \footnote{%
same as above.} of $x_{t}$ given $z$
evolves as%
\begin{equation*}
\dot{p}(z,u)=\sum_{i=1}^{m}G_{i}\left( \lambda _{i0}^{0}+\lambda
_{i0}(z)+\sum_{j=1}^{p}\mu _{ij}(z)u_{j}\right) p(z,u)
\end{equation*}%
which will be abbreviated as 
\begin{eqnarray}
P(z,u) &=&A_{0}+A(z)+\sum_{j=1}^{p}u_{j}B_{j}(z)  \label{CascadeMDPTranProb}
\\
\dot{p}(z,u) &=&P(z,u)p(z,u)
\end{eqnarray}%
The CMDP model is completely specified by $(A_{0},A,B_{j})$.

\textbf{The Admissible Controls, defining $\mathcal{U}$: } The requirements
on $P(z,u)$ to be an infinitesimal generator \textit{for each }$z$ put
constraints on the matrices $A_{0},A,B_{j}$ and impose admissibility
constraints on the controls $u_{j.}$ We will require $A_{0}$ and $A$ to be
infinitesimal generators themselves (for each $t$ \ and $z$) and the $B_{j}$
to be matrices whose columns sum to zero (for each $t$ and $z$). We also
allow the controls to be dependent on $z$ and $x$ which will define the set
of admissible controls $\mathcal{U}$ as the set of measurable functions
mapping the space $\{e_{i}\}_{i=1}^{r}\times \{e_{i}\}_{i=1}^{n}$ to the
space of controls $\mathbb{R}^{p}$ such that the matrix with $j^{th}$ column 
\begin{equation*}
f_{j}=A_{0}e_{j}+A(e_{k})e_{j}+\sum_{i=1}^{p}u_{i}(e_{k},e_{j})B_{j}(e_{k})
\end{equation*}%
for $~j=1..n,~k=1..r$ is an infinitesimal generator. Explicit dependence on $%
t$ is omitted in notation above for clarity.

\subsection{Examples of CMDP}

Two toy examples of CMDP that will be later discussed are outlined below.
Some background on the terminology used in description of portfolio
optimization is in Appendix \ref{SecPOBg}. 

\subsubsection{\label{SecSFPModel}Example 1: A Self-Financing Portfolio Model}

In this toy example on portfolio optimization\footnote{See Appendix C, Section C.1 for some basic definitions on Portfolio Optimization} we will assume that there is one bond and one stock in
the portfolio, with the bond price being fixed at $1$ and the stock having
two possible prices $1$ and $-1/3.$ Thus the price vector takes values in
the set $\{(1,1),(1,-\frac{1}{3})\}$.  Assume a portfolio that can shift
weights between the two assets with allowable weights $W$ of $%
(0,2),(-1,-1),(0,-2)$ so that the portfolio has a constant total position
(of $\frac{-2}{3}$). Further, we allow only weight adjustments of $+1$ or $-1
$ for each asset, and we further restrict the weight shifts to only those
that do not cause a change in net value for any given asset price. The
latter condition makes the portfolio \emph{self-financing}.

The resulting process can be modeled as a cascade MDP. Let $z_{t}$ be the
(joint) prices of the two assets with prices $(1,1)$, $(1,\frac{1}{3})$
represented as states $e_{1},e_{2}$ respectively. Let $x_{t}$ be the choice
of weights with weights $(0,2)$, $(-1,-1)$, $(0,-2)$ represented as states $%
e_{1},e_{2},e_{3}$ respectively. Transition rates of $z_{t}$ are determined
by some pricing model, whereas the rates of $x_{t}$ which represent
allowable weight shifts are controlled by the portfolio manager. The
portfolio value $v(z_{t},x_{t})$ can be written using its matrix
representation, $v(z,x)=z^{T}Vx$, where $V$ is

\begin{equation}
{ V=%
\begin{pmatrix}
2 & -\frac{2}{3} \\ 
-2 & -\frac{2}{3} \\ 
-2 & \frac{2}{3}%
\end{pmatrix}%
}  \label{AssetVMSFP}
\end{equation}%
The portfolio manager is able to adjust the rate $u$ of buying stock (which
has the effect of simultaneously decreasing or increasing the weight of the
bond). The resulting transitions of $x_{t}$ depend on $z_{t}$ (see Figure
\ref{fig:SFP1} ) and transition matrices $P(z)$ of the weights $x_{t}$ can be written
as $P(z)=$\thinspace $A(z)+uB(z)$, where $A(z)$ and $B(z)$are: { {\ 
\begin{equation*}
\begin{tabular}{cc}
$A(e_{1})=\frac{1}{2}%
\begin{pmatrix}
0 & 0 & 0 \\ 
1 & -1 & 1 \\ 
0 & 1 & -1%
\end{pmatrix}%
$ & $A(e_{2})=\frac{1}{2}%
\begin{pmatrix}
-1 & 1 & 0 \\ 
1 & -1 & 0 \\ 
0 & 0 & 0%
\end{pmatrix}%
$ \\ 
$B(e_{1})=%
\begin{pmatrix}
0 & 0 & 0 \\ 
0 & 1 & 1 \\ 
0 & -1 & -1%
\end{pmatrix}%
$ & $B(e_{2})=%
\begin{pmatrix}
1 & 1 & 0 \\ 
-1 & -1 & 0 \\ 
0 & 0 & 0%
\end{pmatrix}%
$%
\end{tabular}%
\end{equation*}%
}} For $P(z)$ to be a proper transition matrix we require admissible
controls $u$ needs to satisfy $\left\vert u\right\vert \leq \frac{1}{2}$%
.~The portfolio manager may choose $u$ in accordance with current values of $%
x_{t}$ and $z_{t}$ so that $u$ is a Markovian feedback controls $%
u(t,z_{t},x_{t}).$ Note that this model differs from the traditional
Merton-like models where only feedback on the total value $v_{t}$ is
allowed. Note that it is the self-financing constraint that leads to the
dependence on the current price $z_{t}$ of the transitions of $x$, which
allows us to model this problem as a cascade.

\begin{center}
\begin{figure}[th]
\centering
\subfigure[$z=e_1$]{
\begin{tikzpicture}[->,>=stealth',shorten >=1pt,auto,node distance=2.4cm,
semithick]
\tikzstyle{every state}=[fill=white,draw=black,thick,text=black,scale=1]
\node[state]         (A)              {\small{(-1,-1)}};
\node[state]         (C) [right of=A] at (A) {\small{(0,-2)}};
\path (A) edge  [bend right] node[below] {$\frac{1}{2}+u$} (C);
\path (C) edge  [bend right] node[above] {$\frac{1}{2}-u$} (A);
\end{tikzpicture}
\label{fig:SFP1}
} 
\subfigure[$z=e_2$]{
\begin{tikzpicture}[->,>=stealth',shorten >=1pt,auto,node distance=2.4cm,
semithick]
\tikzstyle{every state}=[fill=white,draw=black,thick,text=black,scale=1]
\node[state]         (A)              {\small{(-1,-1)}};
\node[state]         (B) [left of=A] {\small{(0,2)}};
\path (A) edge  [bend right] node[above] {$\frac{1}{2}-u$} (B);
\path (B) edge  [bend right] node[below] {$\frac{1}{2}+u$} (A);
\end{tikzpicture}
\label{fig:SFP2}
} \label{fig:SFPModel1}
\caption[Self Financing Portfolio Transition Diagrams]{Transition diagram of
weight $x(t)$ in the self-financing portfolio for various asset prices $z(t)$
are shown in \subref{fig:SFP1} and  \subref{fig:SFP2}.
States $e_{1}$,$e_{2}$ of $z(t)$ correspond to price vectors
(1,1),(1,-1/3) respectively. Self-transitions are omitted for
clarity.}
\end{figure}
\end{center}

\subsubsection{\label{SectionExampleCatsDlm}Example 2: The Cat's Dilemma
Model}

As an example of a cascade MDP, we discuss the cat feeding problem
introduced in Section \ref{SecIntro}. The feeding cat is represented by the
process $x(t)$ with four states $e_{4}=$\emph{Unfed}, $e_{1}=$\emph{Ate Meat}%
, $e_{2}=$\emph{Ate Fish}, $e_{3}=$\emph{Ate Milk}. We assume a constant
feeding rate $f$, and a constant "satisfaction" (digestion) rate $s$ for
each food, upon feeding which the cat always returns to the Unfed state. The
Markov process $z(t)\in \{e_{1},e_{2,}e_{3}\}$ represents availability of
different combinations of food where $e_{1},e_{2},e_{3}$ denote the
combinations \{\emph{Fish,Milk}\},\{\emph{Meat,Milk}\} and \{\emph{Meat,Fish}%
\} respectively. The food provider is unaffected by the cat's eating rate,
and so we can model the process as a cascade with the transition matrix $%
P~ $of $x$ given by (see Figure \ref{fig:subfigureExample}), 
\begin{equation}
P(z,u)=A_{0}+A(z)+B(z)u  \label{CatTransMatrixDef}
\end{equation}%
where the control $u(z,x)$ $\in \lbrack \frac{-1}{2},\frac{1}{2}]$
represents the cat's choice strategy (extreme values $\pm \frac{1}{2}$ %
denoting strongest affinity for a particular food in the combination $z$),
and with $A_{0},A(z)$ and $B(z)$ given by {\tiny 
\begin{equation*}
\begin{tabular}{ccc}
$A_{0}=%
\begin{pmatrix}
-s & 0 & 0 & 0 \\ 
0 & -s & 0 & 0 \\ 
0 & 0 & -s & 0 \\ 
s & s & s & 0%
\end{pmatrix}%
$ &  &  \\ 
$A(e_{1})=%
\begin{pmatrix}
0 & 0 & 0 & 0 \\ 
0 & 0 & 0 & \frac{f}{2} \\ 
0 & 0 & 0 & \frac{f}{2} \\ 
0 & 0 & 0 & -f%
\end{pmatrix}%
$ & $A(e_{2})=%
\begin{pmatrix}
0 & 0 & 0 & \frac{f}{2} \\ 
0 & 0 & 0 & 0 \\ 
0 & 0 & 0 & \frac{f}{2} \\ 
0 & 0 & 0 & -f%
\end{pmatrix}%
$ & $A(e_{3})=%
\begin{pmatrix}
0 & 0 & 0 & \frac{f}{2} \\ 
0 & 0 & 0 & \frac{f}{2} \\ 
0 & 0 & 0 & 0 \\ 
0 & 0 & 0 & -f%
\end{pmatrix}%
$ \\ 
$B(e_{1})=%
\begin{pmatrix}
0 & 0 & 0 & 0 \\ 
0 & 0 & 0 & f \\ 
0 & 0 & 0 & -f \\ 
0 & 0 & 0 & 0%
\end{pmatrix}%
$ & $B(e_{2})=%
\begin{pmatrix}
0 & 0 & 0 & -f \\ 
0 & 0 & 0 & 0 \\ 
0 & 0 & 0 & f \\ 
0 & 0 & 0 & 0%
\end{pmatrix}%
$ & $B(e_{3})=%
\begin{pmatrix}
0 & 0 & 0 & f \\ 
0 & 0 & 0 & -f \\ 
0 & 0 & 0 & 0 \\ 
0 & 0 & 0 & 0%
\end{pmatrix}%
$%
\end{tabular}%
\end{equation*}
}

\begin{center}
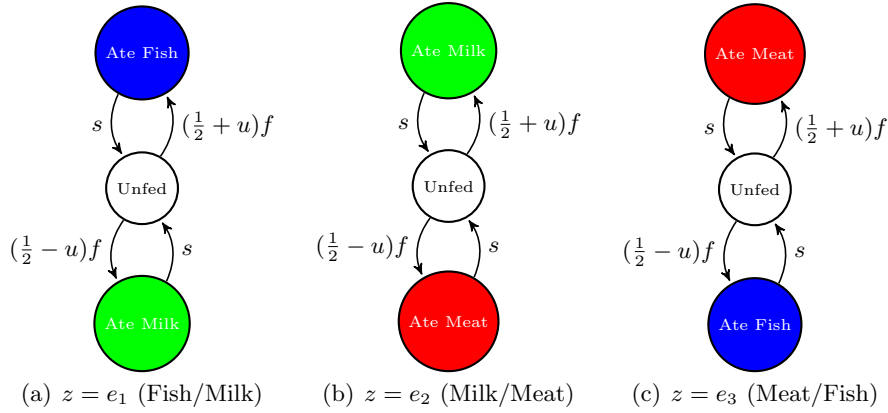
\begin{figure}[t]
\centering
\subfigure[$z=e_1$ (Fish/Milk)]{
\begin{tikzpicture}[->,>=stealth',shorten >=1pt,auto,node distance=1.8cm,
semithick]
\tikzstyle{every state}=[fill=white,draw=black,thick,text=black,scale=1]
\node[state]         (A)              {\tiny{Unfed}};
\tikzstyle{every state}=[fill=blue,draw=black,thick,text=white,scale=1]
\node[state]         (B) [above of=A] {\tiny{Ate Fish}};
\tikzstyle{every state}=[fill=green,draw=black,thick,text=white,scale=1]
\node[state]         (C) [below of=A] at (A) {\tiny{Ate Milk}};
\path (A) edge  [bend right] node[right] {$(\frac{1}{2}+u)f$} (B);
\path (A) edge  [bend right] node[left] {$(\frac{1}{2}-u)f$} (C);
\path (B) edge  [bend right] node[left] {$s$} (A);
\path (C) edge  [bend right] node[right] {$s$} (A);
\end{tikzpicture}
\label{fig:subfig1}
} 
\subfigure[$z=e_2$ (Milk/Meat)]{
\begin{tikzpicture}[->,>=stealth',shorten >=1pt,auto,node distance=1.8cm,
semithick]
\tikzstyle{every state}=[fill=white,draw=black,thick,text=black,scale=1]
\node[state]         (A)              {\tiny{Unfed}};
\tikzstyle{every state}=[fill=green,draw=black,thick,text=white,scale=1]
\node[state]         (B) [above of=A] {\tiny{Ate Milk}};
\tikzstyle{every state}=[fill=red,draw=black,thick,text=white,scale=1]
\node[state]         (C) [below of=A] at (A) {\tiny{Ate Meat}};
\path (A) edge  [bend right] node[right] {$(\frac{1}{2}+u)f$} (B);
\path (A) edge  [bend right] node[left] {$(\frac{1}{2}-u)f$} (C);
\path (B) edge  [bend right] node[left] {$s$} (A);
\path (C) edge  [bend right] node[right] {$s$} (A);
\end{tikzpicture}
\label{fig:subfig2}
} 
\subfigure[$z=e_3$ (Meat/Fish)]{
\begin{tikzpicture}[->,>=stealth',shorten >=1pt,auto,node distance=1.8cm,
semithick]
\tikzstyle{every state}=[fill=white,draw=black,thick,text=black,scale=1]
\node[state]         (A)              {\tiny{Unfed}};
\tikzstyle{every state}=[fill=red,draw=black,thick,text=white,scale=1]
\node[state]         (B) [above of=A] {\tiny{Ate Meat}};
\tikzstyle{every state}=[fill=blue,draw=black,thick,text=white,scale=1]
\node[state]         (C) [below of=A] at (A) {\tiny{Ate Fish}};
\path (A) edge  [bend right] node[right] {$(\frac{1}{2}+u)f$} (B);
\path (A) edge  [bend right] node[left] {$(\frac{1}{2}-u)f$} (C);
\path (B) edge  [bend right] node[left] {$s$} (A);
\path (C) edge  [bend right] node[right] {$s$} (A);
\end{tikzpicture}
\label{fig:subfig3}
} \label{fig:subfigureExample}
\caption[Cat's Dilemma Transition Diagrams]{Transition diagram of cat
feeding states $x(t)$ in the Cat's Dilemma for various food combinations $%
z(t)$ are shown in \subref{fig:subfig1}, \subref{fig:subfig2} and 
\subref{fig:subfig3}. Self-transitions are omitted for clarity.}
\end{figure}
\end{center}

\section{Optimal Control Problem Type I : Expected Utility Maximization}

As alluded to in the introduction, the first category of optimal control
problems on cascade MDPs is one where performance measure is the expectation
of a functional, and hence linear in the underlying probabilities. We will
primarily discuss the fully observable (full feedback), finite time-horizon
case and derive a general solution as a matrix differential equation.

\subsection{\label{OptProbDefSection}Problem Definition}

Fix a finite time horizon $T$ on the cascade MDP $(z_{t},x_{t})$ defined in
Section \ref{SectionCascadeMDPModel1}. and define the cost function $\eta $, 
\begin{equation}
\eta (u)=\mathbb{E}\int_{0}^{T}(z^{T}(\sigma )\mathbf{L}^{T}(\sigma
)x(\sigma )+\psi (u(\sigma ))d\sigma +z^{T}(T)\Phi ^{T}(T)x(T)
\label{OCPIPerfMeas}
\end{equation}%
where $c,\phi $ are real-valued functions on the space $\mathbb{R}^{+}\times
\{e_{i}\}_{i=1}^{r}\times \{e_{i}\}_{i=1}^{n}$, that are represented by the
real matrices $\mathbf{L}(t)$ and $\mathbf{\Phi }(t)$ as $\ c(t,z,x)=z^{T}%
\mathbf{L}(t)x$ and $\phi (t,z,x)=z^{T}\mathbf{\Phi }(t)x$; and $\psi $ a
(Borel) measurable function $\mathbb{R}^{p}\rightarrow \mathbb{R}$. If $c$
is bounded the problem of finding the solution to 
\begin{equation}
\eta ^{\ast }=\min_{u\in \mathcal{U}}\eta (u)  \label{OptCtrlProblem}
\end{equation}%
is well-defined and will be subsequently referred to as Problem (\textbf{%
OCP-I}). The corresponding optimal control is given by 
\begin{equation}
u^{\ast }=\arg \min_{u\in \mathcal{U}}\eta (u)  \label{OptCtrl}
\end{equation}

\subsection{Solution Using Dynamic Programming Principle}

\begin{theorem}
\label{TheoremBellman}Let $(z_{t},x_{t})$ be a cascade MDP\ as defined in
Section \ref{SectionCascadeMDPModel1} with $C,A_{0},A$ and $B_{i}$ as
defined thereof. Let $T>0,$ and $\mathcal{U},\psi $,$\Phi $and $\eta $ be as
defined in section \ref{OptProbDefSection}. Then there exists \ a unique
solution to the equation (on the space of $n\times r$ matrices) 
\begin{eqnarray}
\dot{K} &=&-KC-L-A_{0}^{T}K-A^{T}(z)K-\min_{u(z,x)\in \mathcal{U}%
}(\sum_{i=1}^{p}u_{i}z^{T}K^{T}B_{i}(z)x+\psi (u))~
\label{BellmanPartiallyDec} \\
K(T) &=&\Phi (T)
\end{eqnarray}%
%
%
%
%
%
%
on the interval $[0,T],$ where $A^{T}(z)K$ denotes the matrix whose $%
j^{\prime }th$ column is $A(e_{j})K^{T}e_{j}^{T}$ (which can be more
explicitly written as $\sum_{z}A^{T}(z)Kzz^{T}$, that is, the matrix
representation of the functional $x^{T}A^{T}(z)Kz$). Furthermore, if $K(t)$ %
is the solution to \ref{BellmanPartiallyDec} then the optimal control
problem \textbf{OCP-I }defined in (\ref{OptCtrlProblem}) has the solution 
\begin{eqnarray}
\eta ^{\ast } &=&\mathbb{E}z^{T}(0)K^{T}(0)x(0)~  \label{BellmanOptCtrlSoln}
\\
~u^{\ast } &=&\arg \min_{u(z,x)\in \mathcal{U}}(%
\sum_{i=1}^{p}z^{T}K^{T}u_{i}B_{i}(z)x+\psi (u_{i}))
\end{eqnarray}
\end{theorem}

\begin{proof}
With $z$,$x,\eta $ as defined above let the minimum return function be $%
k(t,z,x)=z^{T}K^{T}(t)x,$where $K(t)$ is an $n\times r$ matrix, so that$%
~k(0,z(0),x(0)=\eta ^{\ast }$. Using Ito rule for$~z^{T}K^{T}x$ 
\begin{equation*}
d(z^{T}K^{T}x)=\sum_{i=1}^{s}z^{T}H_{i}^{T}K^{T}xdM_{i}+z^{T}\dot{K}%
^{T}x+\sum_{i=1}^{n}z^{T}K^{T}G_{i}xdN_{i}
\end{equation*}%
Since the process $dN_{i}-(\lambda _{i0}^{0}+\lambda
_{i0}(z)+\sum_{j=1}^{p}\mu _{ij}(z)u_{j})dt$ is a martingale equating the
expectation to zero gives 
\begin{eqnarray*}
\mathbb{E(}\sum_{i=1}^{n}z^{T}K^{T}G_{i}xdN_{i})&=&\mathbb{E(}g(t,x,z,u)dt) \\
\mathbb{E(}\sum_{i=1}^{s}z^{T}H_{i}^{T}K^{T}xdM_{i})&=&\mathbb{E(}%
z^{T}C^{T}K^{T}x) \\
\end{eqnarray*}%
with $g(t,x,z,u)=z^{T}K^{T}A_{0}x+z^{T}K^{T}A(z)+%
\sum_{i=1}^{p}z^{T}K^{T}u_{i}B_{i}(z)x$.~Writing $c(t,z,x)+\psi
(u)=f(t,z,x,u)$ and $z^{T}C^{T}K^{T}x+g(t,x,z,u)$ $=\xi (t,x,z,u)$, a simple
application of the stochastic dynamic programming principle shows that 
\begin{equation*}
z(t)^{T}\dot{K}(t)^{T}x(t)+\min_{u}(\xi (t,x,z,u)+f(t,z,x,u))\geq 0
\end{equation*}%
The minimum value of $0$ is actually achieved by $u^{\ast }$so that the
inequality above must be an equality. Introducing notation $A^{T}(z)K$, we
get precisely (\ref{BellmanPartiallyDec}). Proof of uniqueness is identical
to that in \cite{BROCKOPT} Theorem 1.
\end{proof}

Note that the Bellman equation (\ref{BellmanPartiallyDec}) is very similar
to that of a single (non cascade) MDP with the additional term $-KC$ %
representing the backward (adjoint) equation for the process $z(t)$ and the
appearance of $z$ in the term for minimization which permits feedback of the
optimal control $u^{\ast }$ on $z$ in addition to $x.$ The matrix $K$ above
is also known as the \textbf{Minimum Return Function.~}The above solution is
a single point boundary value problem instead of two-point.\ For small $KC,$
the above decouples one column at a time. This form is readily generalizable
to multifactor \ MDPs as well.

\begin{corollary}
(\emph{Quadratic Cost of Control}) Under the hypothesis of the above
theorem, if $\psi (u_{i})=u_{i}^{2}$ then if $u_{i}(t,z,x)=\frac{-1}{2}%
z^{T}(t)K^{T}(t)B_{i}(z)x(t)$ lies in the interior of $\mathcal{U}$ then it
is the optimal control. Otherwise the optimal control is on the boundary of $%
\mathcal{U}$. If the former is the case at every $t\in \lbrack 0,T]$, then
equation (\ref{BellmanPartiallyDec}) defining the optimal solution becomes
(where the notation $M^{.2}$ for a matrix is element-wise squared matrix): 
\begin{equation*}
\dot{K}=-KC-L-A_{0}^{T}K-A^{T}(z)K+\frac{1}{4}%
\sum_{i=1}^{p}(B_{i}^{T}(z)K)^{.2}
\end{equation*}
\end{corollary}

\begin{corollary}
(\emph{No Cost of Control) }Under the hypothesis of the above theorem, if $%
\psi (u_{i})=0$ then the optimal control is at the boundary of $\mathcal{U}$%
. If $\mathcal{U}$ is defined as the set $\{$ $-a_{i}\leq \left\vert
u_{i}\right\vert \leq a_{i}\}$ the optimal control is the bang-bang control $%
u_{i}(t,z,x)=-a_{i}\sgn(z^{T}K^{T}(t)B_{i}(z)x)$ and equation (\ref%
{BellmanPartiallyDec}) defining the optimal solution becomes%
\begin{equation*}
\dot{K}=-KC-L-A_{0}^{T}K-A^{T}(z)K+\sum_{i=1}^{p}a_{i}\left\vert
B_{i}^{T}(z)K\right\vert ;
\end{equation*}
\end{corollary}

\subsection{\label{SectionMaxPrinciple}Solution Using The Maximum Principle}

The stochastic control problem \textbf{OCP-I} can be formulated as a
deterministic optimization problem (and hence also an open-loop optimization
problem) using probability densities permitting the application of
variational techniques. While this gives us no particular advantage over the
DPP approach in providing a solution to \textbf{OCP-I} , understanding this
formulation is useful for a broader class of problems.

First we note that for the cascade MDP of Section \ref%
{SectionCascadeMDPModel1} the transition matrices $P(z,u)$ in (\ref%
{CascadeMDPTranProb}) can be written in open-loop form 
\begin{equation}
P_{i}=A_{i}+\sum_{j=1}^{p}B_{ij}D_{ij}  \label{GenrDiagForm}
\end{equation}

where $D_{ij}(u)$ is a diagonal matrix with diagonal $%
[u_{j}(e_{i},e_{1})...u_{j}(e_{i},e_{n})]^{T}$,  $A_{i}=A_{0}+A(e_{i})$,~ $%
B_{ij}=B_{j}(e_{i})$ and $P_{i}(u)=P(e_{j},u)$. Next, we can write evolution
of the marginal probabilities $c_{i}(t)=\Pr \{z(t)=e_{i}\}$ and joint
probabilities $p_{ij}(t)=$ $\Pr \{z(t)=e_{i},x(t)=e_{j}\}$ as the state
equations%
\begin{eqnarray}
\dot{c} &=&Cc  \label{StateEqns} \\
\dot{p}_{i} &=&P_{i}p_{i}+p_{i}\dot{c}_{i}/c_{i}  \notag
\end{eqnarray}

where $p_{i}(t)$ is the vector $[p_{i1}(t)~p_{i2}(t)...p_{in}(t)]^{T},~c(t)$ %
the vector $[c_{1}(t)...c_{r}(t)]^{T}$.

Now we are ready to show the equivalence of the variational approach to the
Bellman approach for the problem (\textbf{OCP-I})

\begin{theorem}
\label{TheoremMaxP}Let $z\in \{e_{i}\}_{i=1}^{r}$ $,x\in
\{e_{i}\}_{i=1}^{n}~ $and $C,A_{0},A(z),B_{i}(z)$ be as defined in Section %
\ref{SectionCascadeMDPModel1} and let the $n-$vectors $p_{i}(t)$ (for $%
i=1..r $)$~ $and $r-$vector $c(t)$ satisfy (\ref{StateEqns}) with $P_{i}$ as
defined by (\ref{GenrDiagForm}) for $D_{ij}$ arbitrary time dependent
diagonal $p-$matrices, considered as controls. Then the minimization of 
\begin{equation*}
\int_{0}^{T}\sum_{i=1}^{r}(e_{i}^{T}\mathbf{L}^{T}p_{i}+\sum_{j=1}^{p}e^{T}%
\psi (D_{ij})p_{i})dt+\sum_{i=1}^{r}e_{i}^{T}\Phi ^{T}p_{i}(T)~
\end{equation*}
for $n\times r$ real matrices $\mathbf{L}$ and $\Phi $ and (Borel)
measurable function $\psi :\mathbb{R}^{p}\rightarrow \mathbb{R}$, subject to
the constraints that $P_{i}\in \mathbb{P}^{n}$ results in a choice for the $%
k^{th}$ element of $D_{ij}$ which equals the optimal control $u_{j}^{\ast
}(e_{i,}e_{k})$ of Theorem \ref{TheoremBellman}.
\end{theorem}

\begin{proof}
Using (\ref{GenrDiagForm}) the Hamiltonian $H$ and costate $(q,s)$ for state
equations (\ref{StateEqns}) for the minimization problem of the theorem
become, assuming normality and stationarity of $z(t)$, are%
\begin{eqnarray}
H
&=&\sum_{i=1}^{r}q_{i}^{T}A_{i}+%
\sum_{j=1}^{p}q_{i}^{T}B_{ij}D_{ij}p_{i}+e^{T}\psi
(D_{ij})p_{i}+l_{i}^{T}p_{i}  \notag \\
\dot{q}_{i}
&=&-(A_{i}+\sum_{j=1}^{p}B_{ij}D_{ij})^{T}q_{i}-l_{i}^{T}-\sum_{j=1}^{p}\psi
(D_{ij})e  \label{Hamiltonian}
\end{eqnarray}
where $l_{i}\equiv e_{i}^{T}L^{T},~\phi _{i}\equiv e_{i}^{T}\Phi ^{T}$ and $%
\psi (D_{ij})\equiv \psi (u_{j}(e_{i},x))$. Introducing minimization of $H$ %
with respect to $D_{ij}$ we see that it is achieved by minimizing $%
\sum_{i=1}^{r}\sum_{j=1}^{p}q_{i}^{T}B_{ij}D_{ij}p_{i}+e^{T}\psi
(D_{ij})p_{i}.$ Noting that $\psi (D_{ij})$ is also diagonal, simple
observation shows that the above is precisely minimized when $%
\sum_{j=1}^{p}(D_{ij}^{T}B_{ij}^{T}q_{i}+\psi (D_{ij})e)$ is minimized for
each $i$ (as $p_{ik}\geq 0$). The maximum priniciple thus gives the
following necessary condition for optimality: 
\begin{equation*}
\dot{q}_{i}=-A_{i}^{T}q_{i}-l_{i}^{T}-\min_{D_{ij}}(%
\sum_{j=1}^{p}D_{ij}^{T}B_{ij}^{T}q_{i}+\psi (D_{ij})e)
\end{equation*}%
We note that since stationarity of $z(t)$ was assumed, the above equation
exactly corresponds to each column of the Bellman matrix equation (\ref%
{BellmanPartiallyDec}) for $K,$ of Theorem \ref{TheoremBellman}. (Note that
the result is valid for non-stationary $z(t)$ as well and algebraic
manipulation shows $(\dot{c}_{i}/c_{i})$ terms to correspond to the $-KC$%
~term in \ref{BellmanPartiallyDec}).
\end{proof}

\begin{remark}
Note that in the variational formulation, linearity of the Hamiltonian in
the state variable $p$ for the problem \textbf{OCP-I~}resulted in complete
decoupling of the state and costate equations $q_{i}$ and $p_{i}$ thereby
permitting an explicit solution identical to that of \ref%
{BellmanPartiallyDec}.~\ However, if we we restrict the set of admissible
controls in to allow feedback on the state $x$ but not on the state $z$ in
problem \textbf{OCP-I~}we get a non-trivial variant problem, a \textit{%
partial feedback} problem, in which case, one can see that in the analysis
in Theorem \ref{TheoremMaxP} the minimization of $\sum_{i=1}^{r}%
\sum_{j=1}^{p}q_{i}^{T}B_{ij}D_{ij}p_{i}+e^{T}\psi (D_{ij})p_{i}$, in
general, does not lead to a decoupling of the state and costate equations.
\end{remark}

\subsection{\label{SecPOEx}Example: A Self-Financing Portfolio}

A toy model of portfolio optimization is discussed as an example or problem~%
\textbf{OCP-I}. Appendix B has a short background on portfolio theory and
also discusses a variety of \textbf{OCP-I }problems on different portfolio
models. Consider the problem of maximizing the expected terminal value $v(T)~
$of the portfolio for a fixed horizon $T$ for the self-financing portfolio
model of Section \ref{SecSFPModel}. With $x,z,u,d,V,A,B,D$ as defined
thereof, we wish to maximize the performance measure given by 
\begin{equation*}
\eta (u,d)=\mathbb{E(}v(T))
\end{equation*}%
Using Theorem \ref{TheoremBellman} we see the solution to this \textbf{OCP-I}
problem is obtained by solving the matrix equation with boundary condition $%
K(T)=-V$ 
\begin{equation}
\dot{K}=-KC-A^{T}(z)K+\frac{1}{2}\left\vert K^{T}B(z)\right\vert +\frac{1}{2}%
\left\vert K^{T}D(z)\right\vert   \label{SelfFnSoln}
\end{equation}%
with the optimum performance measure and controls (in feedback form) given
by 
\begin{eqnarray*}
\eta ^{\ast } &=&z^{T}(0)K^{T}(0)x(0) \\
u^{\ast }(t,z,x) &=&-\frac{1}{2}\sgn(z^{T}K(t)^{T}B(z)x) \\
d^{\ast }(t,z,x) &=&-\frac{1}{2}\sgn(z^{T}K(t)^{T}D(z)x)
\end{eqnarray*}%
with $K(t)$ being the solution to (\ref{SelfFnSoln}). Some solutions for (%
\ref{SelfFnSoln}) and corresponding optimal controls are plotted for $T=15$ %
\ is shown in Figure \ref{fig:FigAssetEx_SF} for various initial conditions
(mixes of the assets in the portfolio initially). Results also show that as $%
T\rightarrow \infty ,$the value of $\eta ^{\ast }$ approaches a constant
value of $1.24$ regardless of the initial values $z(0),x(0).$ That is the
maximal possible terminal value for the portfolio is $1.24.$ However, we do
not see a steady state constant value for the optimal controls $u^{\ast
}(z,x)$ and $d^{\ast }(z,x)$ and that near the portfolio expiration
date, more vigorous buying/selling activity is necessary. \ If the matrix $C$
were reducible or time-varying in our example, multiple steady-states are
possible as $T\rightarrow \infty $ and the initial trading activity will be
more significant.

\begin{figure}[tbp]
\centering
\includegraphics[width=0.8\linewidth]{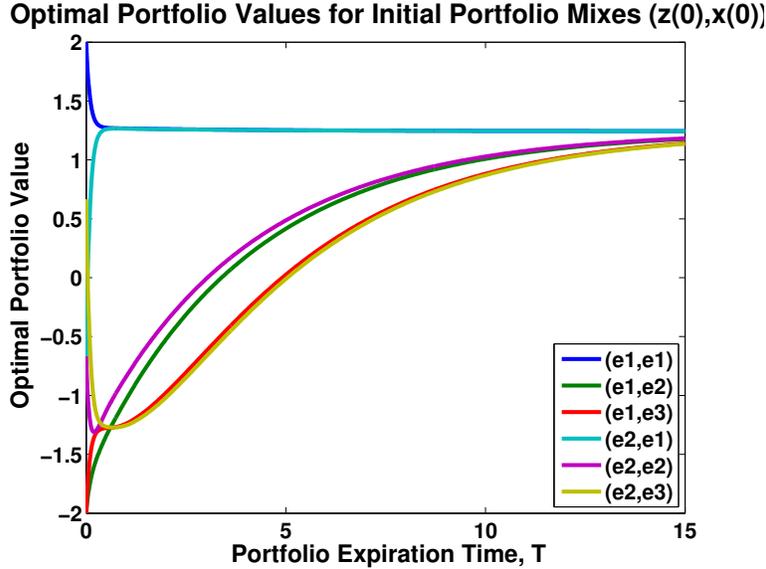}
\caption[Self Financing Portfolio]{Minimum Return Function $k(t,z,x)$,
for the self-financing portfolio with maximal terminal wealth, shown for various
values of $z,x$ specified as $(e_{i},e_{j})$ vectors.%
}
\label{fig:FigAssetEx_SF}
\end{figure}\

%

Two instances of simulation of application of the above optimal controls are
shown in Figures \ref{fig:FigSimulation_1_FF} and \ref{fig:FigSimulation_2_FF}.
 In the first case we see that one is able to
benefit from $x(0)$ being in state $e_{2}$ which is the one that corresponds
to maximal value of the portfolio, but in which state no trading can take
place. We can hold that value and it more than offsets any devaluation due
to stock price decline since the stock is more probable to have a higher
price than lower. In the second simulation, we are unable to achieve state $%
x=e_{2}$, which happens because this state can be attained only in the less
probable case of a lower stock price. However, the optimal strategy still
trues to maximize the portfolio value by forcing state $x=e_{3}$ when the
price is lower, but since this state is less likely, we need only switch to
this sell-out strategy for a small portion of the time. The final value is
most sensitive to the final trading activity. The optimal strategy allows
us to maximize the portfolio value in all cases, and on the average, gives
us the best value.

\begin{figure}[!ht]
\centering
\includegraphics[width=0.9\linewidth]{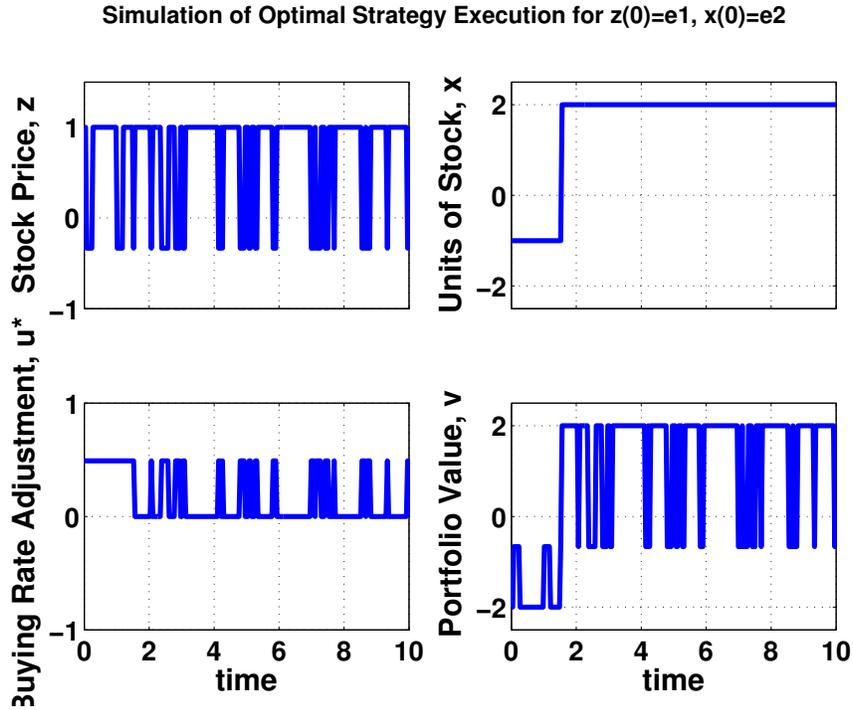}
\caption[Self Financing Portfolio]{ Simulation 1 of optimal control for self-financing portfolio  in 
\ref{SecPOEx}}.%
\label{fig:FigSimulation_1_FF}
\end{figure}

\begin{figure}[!ht]
\centering
\includegraphics[width=0.9\linewidth]{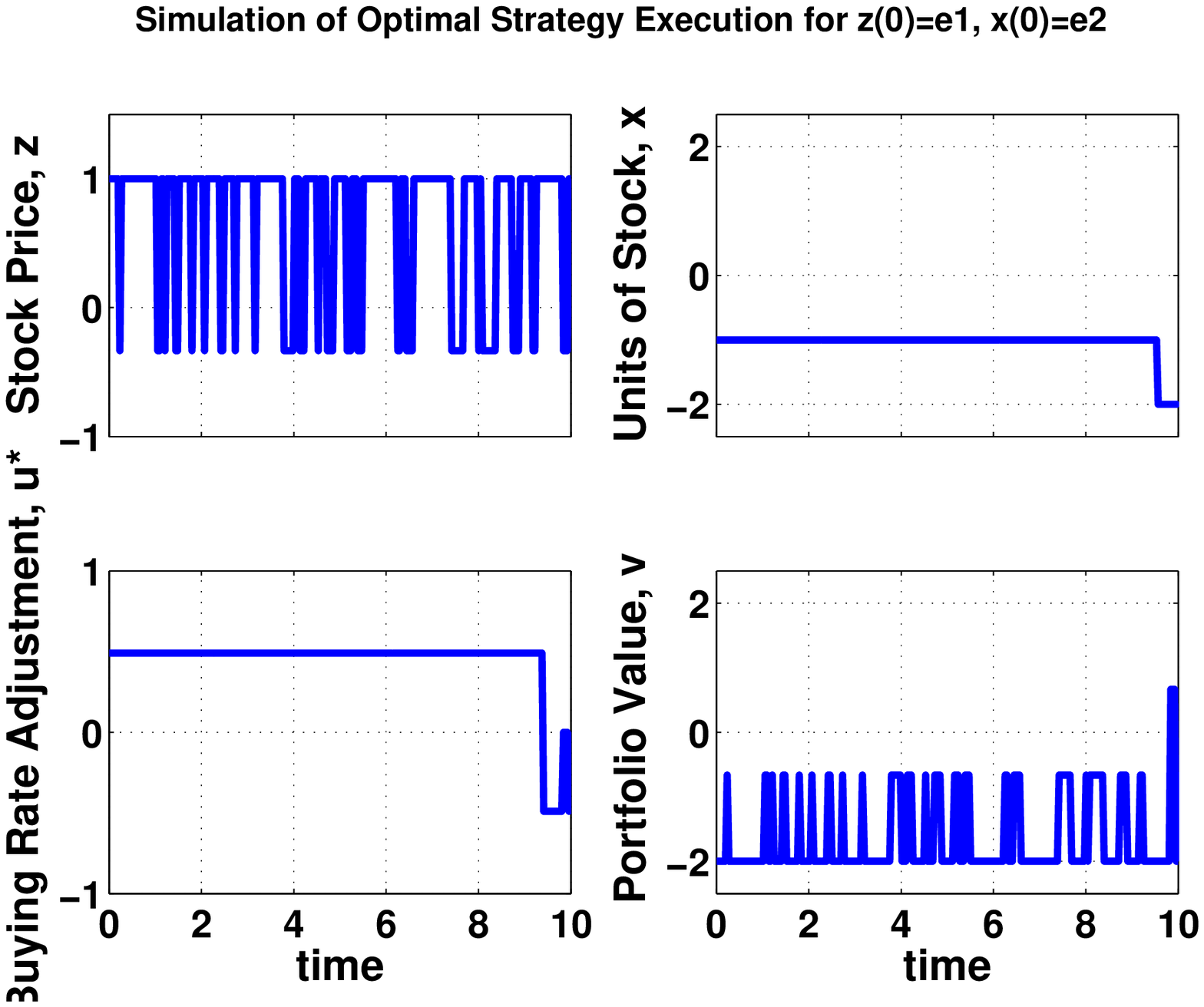}
\caption[Self Financing Portfolio]{ Simulation 2 of optimal control for self-financing portfolio  in 
\ref{SecPOEx}}.%
\label{fig:FigSimulation_2_FF}
\end{figure}

Our approach of using a cascade model is a more realistic model for
portfolio as it is dynamic hedged. Traditional Gauss-Markov models assume
continuous hedging which is unrealistic. Our model can be easily extended to
include features such as transaction costs, etc. Furthermore, by modeling it
as a cascade, we have a computationally scalable solution. The computation
time as a function of the dimensionality of the weight s$z_{t}$ for a
decomposed representation and fully coupled representation (using Bellman
equations on the joint process directly) for various expiration times are
shown in Figure \ref{fig:FigComplexity}. We see that the solution on a coupled state space grows
exponentially with the dimensionality of $z_{t}$ whereas our solution scales
linearly.

\begin{figure}[!ht]
\centering
\subfigure[$T=1$]{
\includegraphics[width=0.3\linewidth]{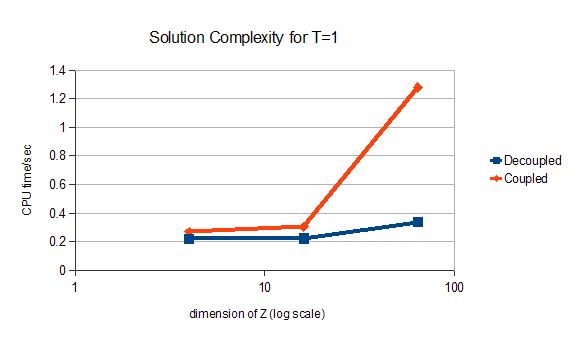}
}
\subfigure[$T=10$]{
\includegraphics[width=0.3\linewidth]{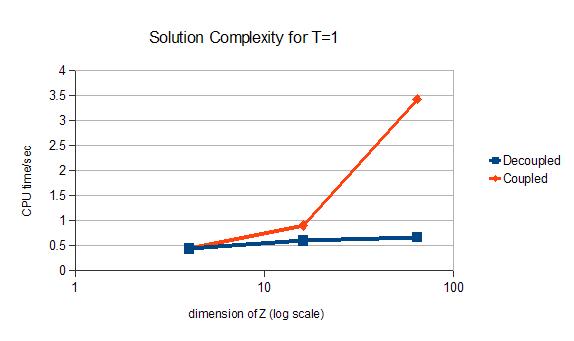}
}
\subfigure[$T=100$]{
\includegraphics[width=0.3\linewidth]{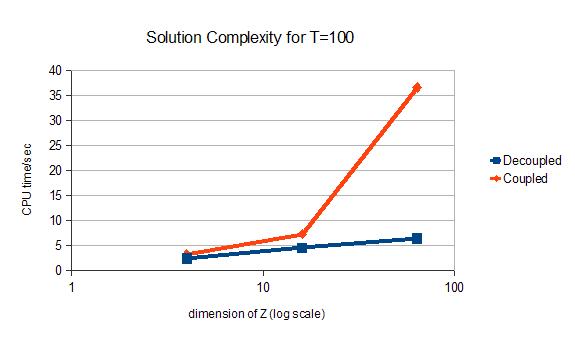}
}
\caption[Computational Complexity]{CPU time in seconds for asset/bond self-financing toy problem when the number of states of $z_{t}$ (possible price combinations) increases, for different expiration times $T=1,10,100$. The decoupled solution scales with dimensionality whereas the coupled solution does not. }
\label{fig:FigComplexity}
\end{figure}

\section{ Optimal Control Problem Type II: Diversification Maximization}

The second category of cascade MDP\ problems are those of optimization of
functionals that are non-linear with respect to the probabilities $p_{ij}$,
such as portfolio diversification or fairness of choices in decision making.
As alluded to in the introduction, these problems are often singular in the
sense that the dynamic programming or maximum principle fail to give a
solution, and we will explore this through an example. In general, this
class of problems falls in the category where the performance measure to be
optimized is a non-linear function of expectation. That is, for a non-linear 
$f$ we want to minimize 
\begin{equation}
\eta (u)=\int_{0}^{T}f(\mathbb{E}(l(t,z_{t},x_{t},u)))dt  \label{NonStdOC}
\end{equation}%
where $l(.)$ is some loss function. For example, $\eta
(u)=\int_{0}^{T}\left( f(x)-\mathbb{E}f(x)\right) ^{2}dt$ minimizes the
variance of function $f$ and $\eta (u)=\int_{0}^{T}\left( \mathbb{E}x_{1}-%
\mathbb{E}x_{2}\right) ^{2}dt$ specifies adherence to a particular state.

\subsection{\label{OptControlOCPIIDef}Quadratic Problem with No Control Cost}

We formulate a problem that is a particular case of (\ref{NonStdOC}). Given a cascade  $(z_{t},x_{t},\mathcal{U})$ with model $(A,B_{i})$ we
define Problem \textbf{OCP-II} as the optimal control problem 
\begin{equation}
\eta ^{\ast }=\min_{u\in \mathcal{U}}\lim_{T\rightarrow \infty }\frac{1}{T}%
\int_{0}^{T}(p_{t}^{T}Qp_{t}+m^{T}p_{t})dt
\end{equation}%
where $Q\geq 0,$ $m$ is a vector and $p_{t}$ is the marginal probability
vector of $x_{t}.$ We note that the stochastic dynamic programming principle
is not directly applicable a problem of the form (\ref{NonStdOC}), and
application of variational techniques at best gives us a two-point boundary
value problem. Even if we did not have a cascade, the functional of the
above form can result in \emph{singular arcs. }To see this heuristically,
consider the optimal control problem on a non-cascade defined as 
\begin{eqnarray}
\dot{p} &=&(A+\sum_{i}u_{i}B_{i})p  \label{NonCasProb} \\
\eta  &=&\lim_{T\rightarrow \infty }\frac{1}{T}\int_{0}^{T}p^{T}Qpdt  \notag
\end{eqnarray}%
with $\mathcal{U=\{}a_{i}\leq u_{i}\leq b_{i}\}.$ The costate and
Hamiltonian equations for this problem are 
\begin{eqnarray*}
\dot{q} &=&-2Qp-A^{T}q-\left( \sum_{i}u_{i}B_{i}^{T}\right) q \\
H &=&q^{T}Ap+\sum_{i}q^{T}u_{i}B_{i}p+p^{T}Qp
\end{eqnarray*}%
so that 
\begin{equation*}
{\small \frac{\partial H}{\partial u_{i}}=q^{T}B_{i}p}
\end{equation*}%
If $q^{T}B_{i}p=0$ for any finite time interval, then we have a singular arc
so that the Hamiltonian provides no useful information. Characterizing
solutions to such singular optimal control problems is notoriously hard. To
see how we can get a singular arc in the case above, consider a
simplification of (\ref{NonCasProb}) with $A+uB$ of the form $%
A+uf_{i}e_{j}^{T}$ for $u\in \lbrack a,b].$For example,  
\begin{equation*}
B={\small 
\begin{bmatrix}
0 & 0 & 0 \\ 
0 & -1 & 0 \\ 
0 & 1 & 0%
\end{bmatrix}%
=%
\begin{bmatrix}
0 \\ 
-1 \\ 
1%
\end{bmatrix}%
\begin{bmatrix}
0 \\ 
1 \\ 
0%
\end{bmatrix}%
^{T}=(-e_{2}+e_{3})e_{2}^{T}=f_{3}e_{2}^{T}~}
\end{equation*}%
In steady state, one can show that 
\begin{equation*}
p(u)=p(0)-\frac{(e_{j}^{T}p(0))u~A^{+}f_{i}}{1+u(e_{j}^{T}A^{+}f_{i})}
\end{equation*}%
where {\small $A^{+}$ is the Moore-Penrose inverse of $A.$ If $\eta (u)%
\mathbf{~}$\textbf{were }of the form $c^{T}p$ (the "usual" stochastic
control case), then $u^{\ast }$ lies on the boundary of  }$\mathcal{U}$. In
the case $\eta (u)$ is of the form $p^{T}Qp$ (i.e. the non-linear stochastic
control case), then it is possible that $u^{\ast }$ is in the interior of $%
\mathcal{U}$. ~For the class of constant controls, if $u^{\ast }\in Int(%
\mathcal{U)}$ then one can show by computation that the corresponding $%
(p(u),q(u)$ correspond to singular arcs. The above argument heuristically
shows why the quadratic control problem OCP-II can be singular.

For a non-cascade, however, the steady state optimal control problem reduces
to a non-functional optimization problem, i.e. that of minimizing $%
p^{T}(u)Qp(u)+m^{T}p(u).$ However, for a cascade, $\eta $ depends on the
marginal probabilities of $x_{t}$ but it is the conditional probabilities $%
x_{t}|z_{t}$ that evolve in accordance with $\dot{p}=Ap.$ In general, it is
difficult to get an expression for $p(u)$ of the steady state marginal
probabilities of $x_{t}$ but we will below consider a special \emph{%
diagonalizable }case where $p(u_{i})$ satisfy $A_{i}p=0$ where $p$
represents the marginal probability vector of $x_{t}.$

\subsection{\label{CatsDilemma}"Cat's Dilemma" Revisited}

In the model presented in Section \ref{SectionExampleCatsDlm}, the
combination of dishes available is random and the cat needs to optimize its
selection strategy so as to get a balance of all three dishes. If we assume $%
s=f=1$ we note that $\mathbb{E}(x_{4})\rightarrow \frac{1}{2}$ as $%
t\rightarrow \infty $ regardless of $z$ or $u.$ Hence, the best balance of
foods is achieved when each of $\mathbb{E}(x_{1}),\mathbb{E}(x_{2}),\mathbb{E%
}(x_{3})$ are as close as possible to $\frac{1}{6}.$ Hence the problem can
be defined as one of minimizing the performance measure 
\begin{equation}
\eta (u)=\lim_{T\rightarrow \infty }\frac{1}{T}\int_{0}^{T}\left\Vert 
\mathbb{E}(Qx(t,z,u))-m\right\Vert ^{2}dt  \label{CatOptFn}
\end{equation}
where $\left\Vert {}\right\Vert $ is the Euclidean norm on $\mathbb{R}^{4}$
and $Q,m$ defined as 
\begin{equation}
Q=%
\begin{pmatrix}
1 & 0 & 0 & 0 \\ 
0 & 1 & 0 & 0 \\ 
0 & 0 & 1 & 0 \\ 
0 & 0 & 0 & 0%
\end{pmatrix}%
;~m=\frac{1}{6}%
\begin{pmatrix}
1 \\ 
1 \\ 
1 \\ 
0%
\end{pmatrix}
\label{CatOptFnQmDef}
\end{equation}

\subsection{\label{SectionBinaryDecisionProblem}A Binary Decision Problem}

The cat's dilemma can be generalized to a class of problems where one needs
to make a choice given two possibilities at a time, so as to maximize the
diversity of outcomes as a result of one's choices. If the total number of
outcomes is $N$ then binary possibilities are represented by the Markov
process $z(t)\in \{e_{i}\}_{i=1}^{r}$ , $r=\frac{1}{2}N(N-1)$, having
generator $C$ where, and the outcomes by the cascade Markov process $x(t)\in
\{e_{i}\}_{i=1}^{n}$,$~n=N+1,$ with transition matrix as in (\ref%
{CatTransMatrixDef}) with $r$ and $n$ dimensional analogs for $A_{0},A(z)$
and $B(z)$. The admissibility set $\mathcal{U}$ of controls $u(t,z,x)$ is
the set of functions $u:\mathbb{R}^{+}\times \{e_{i}\}_{i=1}^{r}\times
\{e_{i}\}_{i=1}^{n}\rightarrow \lbrack \frac{-1}{2},\frac{1}{2}]$ such that
for each $z$ and $t$ the matrix $P(t,z,u)$ is stochastic. (We can generalize
to the situation where, for example, $B(e_{ij})=(e_{i}-e_{j})^{T}e_{n}$ %
where\ $e_{ij}$ is choice $(i,j)$ etc.). This cascade model has simpler
representations as follows. We will assume $s=f=1.$

\begin{proposition}
\label{ProposRepDecProblem}The model described in Section \ref%
{SectionBinaryDecisionProblem} has the following properties.

\begin{enumerate}
\item (Open loop w.r.t $x$ )~For all $t,z$~the dynamics of $x(t)$ do not
depend on the controls $u(t,z,x)$ for all $x\neq e_{n}$

\item (Open loop representation w.r.t $z$ ) There exist rank one matrices $%
A_{j},B_{j}$ of the form $f_{j}e_{n}^{T}$ and (open-loop) controls $u_{j}:%
\mathbb{R}^{+}\rightarrow \lbrack a,b]$ ~for $j=1..r$ such that the
transition matrix (\ref{CatTransMatrixDef})~can be written as 
\begin{equation}
P(t,e_{j},u)=A_{0}+A_{j}+B_{j}u_{j}(t),~\text{for }j=1..r
\label{OpenLoopZCat}
\end{equation}

\item (Triangular Representation) The marginal probabilities $c_{j}(t)=\Pr
\{z(t)=e_{j}\}$ and $p_{k}(t)=\Pr \{x(t)=e_{k}\}$ satisfy the triangular
equations%
\begin{eqnarray}
\dot{c}(t) &=&Cc(t)  \label{DecoupledCat} \\
\dot{p}(t) &=&(A_{0}+\sum_{j=1}^{r}c_{j}(A_{j}+B_{j}u_{j}(t)))p  \notag
\end{eqnarray}
where $p(t)=[p_{1}(t)~...p_{n}(t)~]^{T}$ and $%
c(t)=[c_{1}(t)~...c_{r}(t)~]^{T}$
\end{enumerate}
\end{proposition}

\begin{proof}
Since $B(e_{j})$ is of rank one and of the form $f_{j}e_{n}^{T}$ \ where $%
f_{j}$ is a column vector, the dynamics of $x(t)$ depend only the value of
control in state $x=e_{n}$ and $z.$ Thus w.lo.g write $u(t,z,x)$ as $u(t,z)$
instead. Open loop representation (\ref{OpenLoopZCat}) w.r.t $z$ is made
possible by using the parametrization $u_{j}(t)\equiv u(t,e_{j})$ with $%
B_{j}=B(e_{j})$ and $A_{j}=A(e_{j}).$ The triangular representation follows
from Corollary \ref{CorrSufficientDiagz} (Appendix B) since $%
(A_{j}+u_{j}B_{j})e_{k}=0$ for $j=1..r$ , $k=1..(n-1)$ and that the form of $%
A_{0}$ in (\ref{CatTransMatrixDef}) above above implies that $\Pr (x=e_{n})$ %
is independent of $\Pr (z=e_{j})$ for $j=1..r.$
\end{proof}

The performance measure to maximize diversification of outcomes is (\ref%
{CatOptFn}) which can be written using notation introduced in Proposition %
\ref{ProposRepDecProblem} (with $p(t)$ explicitly written as $p(t,u)$
instead), 
\begin{equation}
\eta (u)=\lim_{T\rightarrow \infty }\frac{1}{T}%
\int_{0}^{T}(Qp(t,u)-m)^{T}(Qp(t,u)-m)dt  \label{CatPerfDet}
\end{equation}%
Two classes of optimal control problems are discussed.

\subsection{ Problem 1 : The Steady State Case}

We assume $z(t)$ to be stationary\footnote{%
If the generator $C$ of $z(t)$ is irreducible then eventually $z(t)$ will
attain a time invariant distribution and hence the solution is no different.}%
. Using the Given a fixed value $T_{0}$ admissibility set $\mathcal{U}%
_{T_{0}}$ is restricted to the set of functions $u_{j}(t),j=1..r$ that are
constant for $t\geq T_{0}$,~as per representation defined in (\ref%
{OpenLoopZCat}). With $\eta (u)$ is in (\ref{CatPerfDet}), the optimization
problem is 
\begin{equation}
\eta ^{\ast }=\min_{u\in \mathcal{U}_{T_{0}}}\eta (u),~\ u^{\ast }=\arg
\min_{u\in \mathcal{U}_{T_{0}}}\eta (u)  \label{SSCatOpt}
\end{equation}
We will call this problem \textbf{OCP-IIS}

\begin{theorem}
\label{CatQuadProg}The solution to the optimization problem \textbf{OCP-IIS}
is given by the solution to the quadratic programming problem%
\begin{equation*}
\eta ^{\ast }=\min_{u}\frac{1}{2}u^{T}Hu+f^{T}u+k~~,\text{subject~to}-\frac{1%
}{2}e\leq Iu\leq \frac{1}{2}e
\end{equation*}%
where $u\in \mathbb{R}^{3},H=\frac{1}{2}A^{T}A,~f=A^{T}b,~k=b^{T}b$ with
matrix $A$ and vector $b$  depending on $(c_{1},c_{2}...c_{r})$~only, and
if $u^{0}$ is the minimizing value for the above, then any function $u(t)$ %
such that $u(t)=u^{0}$ for $t\geq T_{0}$ is an optimal control $u^{\ast }$.
\end{theorem}

\begin{proof}
The infinitesimal generator $X(u)$ for $x(t)$ defined in (\ref{DecoupledCat}%
) is irreducible. Writing the unique time invariant solution to as $p(u)$ a
routine calculation shows that 
\begin{equation}
p(u)=(ee^{T}+X^{T}(u)X(u))^{-1}e  \label{SSProbMP}
\end{equation}%
For $u\in \mathcal{U}_{T_{0}}$ we can write 
\begin{eqnarray*}
\eta (u) &=&\lim_{T\rightarrow \infty }\frac{1}{T}%
(\int_{0}^{T_{0}}(Qp(t,u)-m)^{T}(Qp(t,u)-m))dt \\
& & +\lim_{T\rightarrow \infty }%
\frac{1}{T}\int_{T_{0}}^{T}(Qp(u)-m)^{T}(Qp(u)-m))dt \\
&=&(Qp(u)-m)^{T}(Qp(u)-m)
\end{eqnarray*}%
Since the first integrand $(Qp(t,u)-m)^{T}(Qp(t,u)-m))$ is bounded and the
second integrand $(Qp(u)-m)^{T}(Qp(u)-m))$ is independent of $t$. \ Using (%
\ref{SSProbMP}) write $(Qp(u)-m)^{T}(Qp(u)-m)=\widetilde{p}^{T}\widetilde{p}%
~ $where $\widetilde{p}=\frac{1}{2}Au+B$ \ and $A,b$ are per the statement.
Expanding $\widetilde{p}^{T}\widetilde{p}$ we get the quadratic programming
equation.
\end{proof}

\begin{claim}
The quadratic programming equation (Theorem \ref{CatQuadProg}) has a
solution $\eta ^{\ast }=0$ if and only if the corresponding minimizing value 
$u^{0}$ lies in the interior of the hypercube $[-\frac{1}{2},\frac{1}{2}%
]^{r} $
\end{claim}

\begin{remark}
\label{CatSSSoln}Theorem \ref{CatQuadProg} can also be proved using explicit
computation for the Cat's Dilemma, with $X(u)$ and $p(u)$ given by 
\begin{eqnarray*}
X(u)&=&\allowbreak 
\begin{bmatrix}
-1 & 0 & 0 & c_{3}\left( u_{3}+\frac{1}{2}\right) -c_{2}\left( u_{2}-\frac{1%
}{2}\right) \\ 
0 & -1 & 0 & c_{1}\left( u_{1}+\frac{1}{2}\right) -c_{3}\left( u_{3}-\frac{1%
}{2}\right) \\ 
0 & 0 & -1 & c_{2}\left( u_{2}+\frac{1}{2}\right) -c_{1}\left( u_{1}-\frac{1%
}{2}\right) \\ 
1 & 1 & 1 & -1%
\end{bmatrix}%
\\
p(u)&=&\frac{1}{2}%
\begin{bmatrix}
c_{3}\left( u_{3}+\frac{1}{2}\right) -c_{2}\left( u_{2}-\frac{1}{2}\right)
\\ 
c_{1}\left( u_{1}+\frac{1}{2}\right) -c_{3}\left( u_{3}-\frac{1}{2}\right)
\\ 
c_{2}\left( u_{2}+\frac{1}{2}\right) -c_{1}\left( u_{1}-\frac{1}{2}\right)%
\end{bmatrix}%
\end{eqnarray*}
\end{remark}

and $A,b$ thus being computed as 
\begin{equation*}
A=%
\begin{bmatrix}
0 & -c_{2} & c_{3} \\ 
c_{1} & 0 & -c_{3} \\ 
-c_{1} & c_{2} & 0%
\end{bmatrix}%
,~b=%
\begin{bmatrix}
-\frac{1}{6}+\frac{1}{4}(c_{3}+c_{2}) \\ 
-\frac{1}{6}+\frac{1}{4}(c_{1}+c_{3}) \\ 
-\frac{1}{6}+\frac{1}{4}(c_{2}+c_{1})%
\end{bmatrix}%
\end{equation*}

\begin{remark}
We can solve the quadratic programming explicitly. The solutions $u^{0}\in 
\mathcal{C}$ where $\mathcal{C}$ is the closed cube $[-\frac{1}{2},\frac{1}{2%
}]^{3}$. For the general case of dimensions $r$ and $n$ the results are
similar.
\end{remark}

\label{QuadSolnCatCase1}\textbf{Case 1:~}When $0<c_{j}\leq \frac{2}{3},$ $%
j=1..3.$In this case,$\eta ^{\ast }=0$ and optimal solutions $u^{0}$ are
given by the lines $u_{1}=\frac{1}{2c_{1}}\left( c_{2}+2c_{3}u_{3}+\frac{2}{3%
}\right) ,u_{2}=\frac{1}{2c_{2}}\left( -c_{1}+2c_{3}u_{3}+\frac{2}{3}\right)
$ in the interior of $\mathcal{C}$. \textbf{Case 2: }When $c_{j}\leq \frac{2%
}{3}$ for all $j,$ and $c_{j}=0$ for some $j,j=1..3.$In this case $\eta
^{\ast }=0$ and the solutions are given by, for example, in the case $\{$%
\textbf{\ }$c_{1}=0,c_{3}\leq \frac{2}{3}$ and $c_{2}\leq \frac{2}{3}\}$ the
set of lines $u_{3}=-\frac{1}{3c_{3}}+\frac{1}{2},u_{2}=\frac{1}{3c_{2}}-%
\frac{1}{2}$in the interior of $\mathcal{C}$ but parallel to the faces. 
\textbf{Case 3:~}When $\frac{2}{3}<c_{j}\leq 1$ for some $j.$ Since $H$ is
singular, several local minima may exist. However, the isolines of global
minima are attained along constant values of $c_{i}$ in the case of $\frac{2%
}{3}<c_{i}\leq 1$ and the minimal values increase from $0$ for $c_{i}=\frac{2%
}{3}$ to \ $0.0408$\ \ for $c_{i}=1.$ For example, if $c_{2}>0$ then at most
two global minima are attained at $u=$ $(0,-\frac{1}{2},\frac{1}{2})$ or $%
u=(0,\frac{1}{2},\frac{1}{2})$ \ i.e. on the edges of $\mathcal{C}$. If $%
c_{2}=0$ then the global minimum is attained on the line $u_{1}=0,u_{3}=%
\frac{1}{2}$

\subsection{Problem 2: The Time Varying Case}

Again, we assume $z(t)$ to be stationary. With the admissibility set $%
\mathcal{U}$ is set of functions $u_{j}(t),j=1..r$ such that $u_{j}(t)\in
\lbrack -\frac{1}{2},\frac{1}{2}].$ As per representation defined in (\ref%
{OpenLoopZCat}) and with $\eta (u)$ is in (\ref{CatPerfDet}) the problem is%
\begin{equation}
\eta ^{\ast }=\min_{u\in \mathcal{U}}\eta (u),~~\ u^{\ast }=\arg \min_{u\in 
\mathcal{U}}\eta (u)  \label{CatOptProblem}
\end{equation}
which we call \textbf{OCP-IIT}. In the cases where the steady state optimal
control lies in the interior of $\mathcal{U}$, these controls are also
optimal within the class of time-varying controls.

\begin{proposition}
In the cases described in example of Section \ref{CatSSSoln} where the
solution $u^{0}$ to the quadratic programming equation (Theorem \ref%
{CatQuadProg}) lies in the interior of the hypercube $[-\frac{1}{2},\frac{1}{%
2}]^{r}$ the solution defined in Proposition \ref{CatQuadProg} to \textbf{%
OCP-IIS} for any $T_{0}$ is also a solution to \textbf{OCP-IIT}.
\end{proposition}

\begin{proof}
In the cases of the example of Section \ref{CatSSSoln} where the optimal
controls are in the interior, optimal performance measure is $\eta ^{\ast
}=0.$ \ Since the performance measure $\eta $ defined in (\ref{CatOptFn})
always satisfies $\eta \geq 0,$ thus in these cases a constant control is
also optimal within the class of time-varying controls. And this holds for
constant controls in the class $\mathcal{U}_{T_{0}}$ for any finite $T_{0}$ %
(and thus by no means unique).
\end{proof}

\subsection{Singularity Of Optimal Controls}

The problems in Section \ref{SectionBinaryDecisionProblem} belong to the
category of singular control, and an analysis of singularity of optimal
solutions presents a slightly more general approach to finding the solution
to the time-varying problem (\ref{CatOptProblem}) than the approach above.
For this problem, using the representation of Proposition \ref%
{ProposRepDecProblem}, the Hamiltonian, state and costate equations can be
written as 
\begin{eqnarray}
H
&=&(Qp-m)^{T}(Qp-m)+q^{T}(A_{0}+\sum_{j=1}^{r}c_{j}A_{j}+%
\sum_{j=1}^{r}c_{j}u_{j}B_{j})p  \label{CatMaxPrinc} \\
~~\ \dot{p} &=&(A_{0}+\sum_{j=1}^{r}c_{j}(A_{j}+B_{j}u_{j}(t)))p \\
\dot{q} &=&-2(Qp-m)-(A^{T}+\sum_{j=1}^{r}c_{j}A_{j}^{T})q-(%
\sum_{j=1}^{r}c_{j}u_{j}B_{j}^{T})q
\end{eqnarray}%
However, we see from (\ref{CatMaxPrinc}) that the costate and state
equations are no longer decoupled, and thus trajectories $(q,p)$ that
minimize the Hamiltonian can not simply be obtained by solving an equivalent
minimization of the individual costate/state equations. In fact, as shown
below, we have the case of \textbf{singular arcs}, that is, trajectories
(solutions) where $q^{T}B_{i}p~=0.$ Such trajectories fail to give a
minimization condition for $H$ with respect to $u_{i}.$In such cases, the
Maximum Principle at best remains a necessary condition failing to provide
the optimal solution. Controls $u_{i}$ such that the corresponding solutions 
$(p,q)$ to the state/costate equations form singular arcs will be called 
\textbf{singular controls}.

\begin{proposition}
For $t>T_{0}$, the solutions $u^{\ast }$to the optimal control problem 
\textbf{OCP-IIS }that lie in the interior of $\mathcal{U}$ are singular.
\end{proposition}

\begin{proof}
As $T\rightarrow \infty ,~u^{\ast }$is a constant control and so\ $p$
reaches an invariant distribution. Since the optimal trajectory must satisfy
the state/costate equation, we see that $\dot{q}$ must be zero as well.
Thus, from (\ref{CatMaxPrinc}) we get by putting $X(u)=%
\sum_{j=1}^{r}(A_{0}+c_{j}A_{j}+c_{j}u_{j}B_{j})$ 
\begin{equation*}
-2(Qp-m)-X^{T}(u)q=0
\end{equation*}
Expanding the above for the first $(n-1)$ rows of $X^{T}(u)q$ we get the
equations $q_{n}-q_{i}=-2(p_{i}-\frac{1}{2N})$ for $i=1..n-1.$ These give us
the equations $q_{i}-q_{j}=2(p_{i}-p_{j})$ for $i,j=1..n-1.$ The singularity
conditions $q^{T}B_{i}p=0$ expand to, by putting in the steady value of $%
p(u),$ to $q_{i}-q_{j}=0$ for $i,j=1..n-1.$ Since $p_{i}=p_{j}=\frac{1}{2N}$
when $u^{\ast }$ is in the interior of $\mathcal{U}$ we see that the optimal
solutions are singular.
\end{proof}

Thus, in the steady state case, optimal trajectories are singular. We now
show that this is also the case for the time-varying case.

\begin{proposition}
\label{SingularHzero}For the problem \textbf{OCP-IIT}, the value of the
Hamiltonian on singular arcs is zero.
\end{proposition}

\begin{proof}
The state/costate/Hamiltonian are given by (\ref{CatMaxPrinc}). Without loss
of generality, let $p(0)=e_{n}.$ The state equations can be solved
explicitly for $p_{n}$ using $\dot{p}_{n}=1-2p_{n}$ to yield $p_{n}(t)=\frac{%
1}{2}(1+e^{-2t}).$ Singular arcs satisfy $q^{T}B_{j}p=0$ which expands to $%
p_{n}(q_{i}-q_{j})=0$ for $i,j=1..(n-1)$ i.e. $q_{i}=q_{j}$ for $i,j=1..n-1.$
From $q_{i}(\infty )=0$ we get $\dot{q}_{i}=\dot{q}_{j}$ or $p_{i}=p_{j}$ %
for $i,j=1..n-1$ using the costate equation. Using $\sum_{i=1}^{n}p_{i}=1$
we get the solution $p_{i}(t)=\frac{1}{2N}(1-e^{-2t})$ for $i=1..(n-1).$ Now
plugging these into the costate equations we can explicitly solve for $q_{i}$%
, $i=1..n$ for terminal condition $q_{i}(\infty )=0.$ Omitting details,
plugging the solutions into the Hamiltonian, it can be readily seen that $%
H=0.$
\end{proof}

\begin{corollary}
The solutions $u^{\ast }$to the optimal control problem (\ref{CatOptProblem}%
), such that $\lim_{t\rightarrow \infty }$ $u^{\ast }(t)$ lies in the
interior of $\mathcal{U}$, are singular.
\end{corollary}

\begin{proof}
In steady state, we see that the optimal trajectories (for which $u$ is in
the interior of $\mathcal{U}$)~yields $H=0$ since $(Cp-m)^{T}(Cp-m)=0$ and $%
X(u)p=0$. From the Maximum Principle, this must be the minimizing value of $%
H$ and since there is no explicit dependence of $H$ on $t$ this must be the
value of $H$ on optimal trajectories at all times. Hence, singular
trajectories that satisfy the state/costate equations also minimize the
Hamiltonian and so the entire optimal trajectory is singular.
\end{proof}

Now we show that singular solutions are also optimal for the case when
optimal controls are in the interior of $\mathcal{U}.$

\begin{proposition}
\label{ProposSingularIsOpt}For the problem \textbf{OCP-IIT} the value of $%
\eta $ as defined in (\ref{CatPerfDet}) on singular arcs is zero.
\end{proposition}

\begin{proof}
As in the proof of proposition \ref{SingularHzero}, on singular arcs, $\frac{%
\partial H}{\partial u_{j}}=q^{T}B_{j}p=0$ for $j=1..r$ give the conditions $%
q_{i}=q_{j}$ for $i,j=1...(n-1).$ Evaluating $\frac{d}{dt}(\frac{\partial H}{%
\partial u_{j}})$ for $j=1..r$ \ and setting this to zero (details omitted)
yields further the conditions $p_{i}=p_{j}$ for $i,j=1...(n-1).$ Next,
evaluating $\frac{d^{2}}{dt^{2}}(\frac{\partial H}{\partial u_{j}})$ for $%
j=1..r$ and setting this to zero yields the same equations as in Case 1 and
Case 2 of (a generalized version of) the example presented in Section \ref%
{CatSSSoln} .That is, the equations corresponding to $%
(Qp(u)-m)^{T}(Qp(u)-m)=0$ where $p(u)$ is given by (\ref{SSProbMP}). That
is, $\eta =0.$
\end{proof}

Note that due to the singular nature of the problem, the above analysis does
not give us any information about the optimal control $u^{\ast }.$ However,
we saw from the steady state analysis that a $u$ such that is a constant
value satisfying the quadratic programming problem (QPP) (Proposition \ref%
{CatQuadProg}) after some finite time is an optimal solution. So if we
initially start on a singular trajectory then we remain on it. Otherwise
since $u$ is bounded, we can't jump immediately to the singular trajectories
and so it will be a bang/bang like control till we transition to an optimal
trajectory (though not necessarily constant) control - however, eventually
this will become constant. Thus any control that becomes the constant value
that is the solution to the QPP in finite time, and one that eventually
steers the system onto a singular trajectory is an optimal control.

\section{Conclusion}

A framework for studying the class of problems where the dynamics of a
controllable continuous-time finite-state Markov chain are dependent on an
external stochastic process was introduced in this paper and two categories
of optimal control problems were discussed. In the "type I" or "expected
utility maximization" problems, techniques based upon dynamic programming
were used to provide solutions for a general class of problems in the form
of a matrix differential equation. This result, proved in Theorem \ref%
{TheoremBellman} using the stochastic dynamic programming, was alternately
derived as using the maximum principle, and in the process we were able to
see a more general applicability of the variational approach. These
solutions were applied to a variety of toy examples in the area of dynamic
portfolio optimization. Our factored solutions reduce storage requirements
as well as computational complexity significantly. For example, in our
representation, a coupled problem with $r=10,n=1000$ that would normally
require storing a $1000\times 1000$ matrix needs at most ten $100\times 100$
matrices, thereby providing a reduction by a factor of $10.$ This approach
is also generalizable to multi-factor processes, with many interacting
Markov chains and with even synchronizing transitions. 

Another category of problems, called "type II" or "diversification
maximization" problems with performance functionals that are non-linear in
underlying state probabilities was discussed in the context of a cat feeding
example. It was shown that this problem is singular in the sense that the
maximum principle fails to provide an optimal solution, and alternative
techniques were explored in the solution of this problem.

Ongoing and future work in this area is focused on general techniques for
such singular problems, and extending the class of problems to more complex
ones such as multi-cascades (a set of multiple inter-dependent Markov
chains), hybrid cascades (for instance, a discrete-state Markov chain with
dependencies on continuous-state Gauss-Markov processes) and even decision
processes in the context of quantum Markov chains or quantum controls.
Computational considerations for large scale versions of the toy portfolio
examples presented in this paper will also be investigated.

\bigskip In this paper only the singular control problem defined in Section %
\ref{SectionBinaryDecisionProblem} was analyzed. The general problem of
minimizing a performance measure of the form 
\begin{equation*}
\eta =\int_{0}^{T}(\frac{1}{2}p^{T}Qp+c^{T}p)dt+\frac{1}{2}%
p^{T}(T)S_{f}p(T)+\phi _{f}^{T}p
\end{equation*}%
on a cascade MDP where $Q,S_{f}\geq 0$ needs to be investigated. For the
time-invariant case, following the analysis in where it was shown that if a
minimizer of $c^{T}p$ is in the interior of the admissibility set $\mathcal{U%
}$ then it must define a singular arc, we would like to derive a similar
result for the above case. We would further like to derive, for the
time-invariant case, sufficient conditions for singular arcs to be optimal
(i.e. analog of Proposition \ref{ProposSingularIsOpt}).

Future work in this class of singular problems also involves other
techniques such as variable transformations, as in \cite{Geerts1989135}, the
method of singular perturbations (as in \cite{GRSMN}), and numerical methods
such as Chebyshev-point collocation techniques.


\appendix

\section{Markov Processes on Product State Spaces}\label{SectionProductMP}

We explore representations of a Markov Process $y_{t}~$that evolves on the
product state space $\{e_{i}\}_{i=1}^{r}\times \{e_{i}\}_{i=1}^{n}$.~The
sample path $y(t)$ can be written as the tuple $(z(t),x(t))$ where $z(t)\in
\{e_{i}\}_{i=1}^{r}$ and $x(t)\in \{e_{i}\}_{i=1}^{n}.$ The corresponding
stochastic processes $z_{t}$ and $x_{t}$are the \textbf{components}\emph{\ }%
of $y_{t}.$ The transition matrix for $x_{t}$ may depend on $z(t)~$and hence
describes the propagation of the \emph{conditional probability} distribution 
$p_{x|z}$: The dynamics of component \emph{marginal} probabilities are not
necessarily governed by a single stochastic matrix. Different degrees of
coupling between $x_{t}~$and $y_{t}~$leads to a possible categorization of
the joint Markov Process $y_{t}.$

\begin{definition}
A Markov process $y_{t}~$on the state space $\{e_{i}\}_{i=1}^{r}\times
\{e_{i}\}_{i=1}^{n}$is called \textbf{tightly coupled} or \textbf{%
non-decomposable }if there exist states $(e_{i},e_{j})$ and $(e_{k},e_{l})$
with $i\neq k~$and $j\neq l~$having non-zero transition probability. If all
non-zero transition probabilities are between states of the form $%
(e_{i},e_{j})$ to $(e_{i},e_{k})$, or $(e_{i},e_{j})$ to $(e_{l},e_{j})$
then $y_{t}$ is called \textbf{weakly-coupled }or \textbf{decomposable}.
\end{definition}

\begin{definition}
A decomposable chain on $\{e_{i}\}_{i=1}^{r}\times \{e_{i}\}_{i=1}^{n}$ %
where the transition probability from state $(e_{i},e_{j})$ to $%
(e_{l},e_{j})$ does not depend on $j,$ for all  $i,l,j$ where $1\leq i,l\leq
r$ and $1\leq j\leq n,$ is called a \textbf{Cascade Markov process\footnote{%
In this paper we mainly focus on Cascade Markov processes, and they are
closely related to Markov-modulated Poisson processes (MMPPs) which have
vast applications in traffic control, operations research and electronics
and communications.}}.
\end{definition}

\begin{definition}
A cascade Markov process on $\{e_{i}\}_{i=1}^{r}\times \{e_{i}\}_{i=1}^{n}$ %
where the transition probability from state $(e_{i},e_{j})$ to $%
(e_{i},e_{k})$ does not depend on $i,$ for all$ i,j,k$ where $1\leq i\leq
r~~ $and $1\leq j,k\leq n,$ is called an \textbf{Uncoupled Markov Process}.
\end{definition}

Thus, in a decomposable chain, the jumps in the two component processes are
uncorrelated. However, the rates of the counters (and hence transition
probabilities) in a component can depend on the state of the another
component. In a cascade chain, the rates of the first component ($z_{t})$ do
not depend on the second component $x_{t}.$ In an uncoupled chain, the
component processes $z_{t}$ and $x_{t}$ are completely independent.
Decomposable Markov chains have \emph{functional }transition rates, that is,
the transition rates are state dependent but do not have any synchronous
transitions. Non-decomposable Markov chains exhibit \textit{synchronous
transitions}: that is, transitions amongst states of $x_{t}$ and $z_{t}$ can
occur simultaneously. 

\subsection{Sample\ Path and Transition Probability Representations}

It will be convenient to represent~sample paths $y(t)$ using the Kronecker
tensor product $y(t)=z(t)\otimes x(t)$ instead of the tuple $(z(t),x(t))$.
The state set $y(t)$ then becomes standard basis for $\mathbb{R}^{r\times n}$%
. Following the model in (\ref{ItoMC}) sample paths $y(t)$ have the Ito
representation 
\begin{equation}
dy=\sum_{i=1}^{q}G_{i}ydN_{i}
\end{equation}
\label{JointItoKron}

where $G_{i}~\in \mathbb{G}_{{}}^{rn}$ are distinct. Correspondingly, the
infinitesimal generator $P\in \widehat{P}_{rn}$ can be written as $%
P=\sum_{i=1}^{q}G_{i}\lambda _{i}$ where $\lambda _{i}$ is the rate of
counter $N_{i}.$The following results relate decomposability of sample path
and transition probability representations to the various levels of
couplings defined above.

\begin{proposition}
\label{ItoRepMCPropos}Let the Markov process $y_{t}$ be defined on the state
space $\{e_{i}\}_{i=1}^{rn}$ with the Ito representation (\ref{JointItoKron}%
). Then for each (distinct) $G_{i},$ $i=1..q$ (see notation defined in
Appendix \ref{NotSumm}),

\begin{enumerate}
\item $y_{t}$ is a decomposable Markov process if and only if $G_{i}$ can be
written as either $G_{i}=E_{i}^{r}\otimes G_{i}^{n}$ or $G_{i}=G_{i}^{r}%
\otimes E_{i}^{n}.$

\item If $G_{i}$ can be written as $G_{i}=E_{i}^{r}\otimes G_{i}^{n}$ or $%
G_{i}=G_{i}^{r}\otimes I_{n}$  then $y_{t}$ is a cascade Markov process.

\item If $G_{i}$ can be written as $G_{i}=I_{r}\otimes G_{i}^{n}$ or $%
G_{i}=G_{i}^{r}\otimes I_{n}$  then $y_{t}$ is an uncoupled Markov process
\end{enumerate}
\end{proposition}

\begin{proof}
\begin{enumerate}
\item To prove sufficiency, write (\ref{JointItoKron}) as $%
dy=\sum_{i=1}^{q_{1}}(E_{i}^{r}\otimes G_{i}^{n}~)(z\otimes
x)dN_{i}+\sum_{j=q_{1}+1}^{q}(G_{j}^{r}\otimes E_{j}^{n})(z\otimes
x)dN_{j}~.$ Since $(E_{i}^{r}\otimes G_{i}^{n}~)(z\otimes
x)=E_{i}^{r}z\otimes G_{i}^{n}x$ is a rank one tensor, $zx^{T}$is a rank one
matrix with exactly one non-zero row. Thus jumps in $N_{i}$ change $x$ but
not $z.$ Conversely, jumps in $N_{i}$ that change both $x$ and $z$ must have 
$d(zx^{T})$ of rank $>1$, i.e. $G_{i}\neq E_{i}^{r}z\otimes G_{i}^{n}x$ for
any $E_{i}^{r}$ and $G_{i}^{n}.$

\item In the decomposable change, transitions that change $z$ but not $x$ %
correspond to terms such as $(G_{j}^{r}\otimes I_{n})(z\otimes
x)dN_{j}~=(G_{j}^{r}z\otimes I_{n})dN_{j}$.~Thus the transition $G_{j}^{r}$
is driven by $N_{j}$ only, regardless of $x.$ Since $G_{j}^{r}$ are distinct
for distinct $j$, each transition in $z$ is independent of the value of $x$.

\item Follows by repeating the argument of (2) for the terms $(I_{r}\otimes
G_{i}^{n}~)(z\otimes x)dN_{i}$
\end{enumerate}
\end{proof}

\begin{proposition}
\label{ProposTransRep}Let the Markov process $y_{t}$ be defined on the joint
state space $\{e_{i}\}_{i=1}^{rn}$ with infinitesimal generator $P.$ Then,
as per notation defined in Appendix \ref{NotSumm},

\begin{enumerate}
\item If $y_{t}$ is decomposable, then $P$ can be written in the form $%
P=\sum_{i=1}^{p_{1}}E_{i}^{r}\otimes
B_{i}^{n}+\sum_{i=1}^{p_{2}}B_{i}^{r}\otimes E_{i}^{n}$ where $%
B_{i}^{n},B_{i}^{r}$ are matrices such that $\sum_{i=1}^{p_{1}}B_{i}^{n}\in 
\widehat{P}_{n}$ and $\sum_{i=1}^{p_{2}}B_{i}^{r}\in \widehat{P}_{r}.$

\item If $y_{t}$ is a cascade Markov process then $P$ can be written as $%
P=\sum_{i=1}^{p}E_{i}^{r}\otimes B_{i}^{n}+C\otimes I_{n}$ where $C\in 
\widehat{P}_{r},$ where $B_{i}^{n}$are matrices such that $%
\sum_{i=1}^{p_{1}}B_{i}^{n}\in \widehat{P}_{n}$

\item If $y_{t}$ is an uncoupled Markov process then $P$ can be written as 
\begin{equation}
P=I_{r}\otimes A+C\otimes I_{n}  \label{DiagnolizableMC}
\end{equation}

where $A\in \widehat{P}_{n}$ and $C\in \widehat{P}_{r}$ .
\end{enumerate}
\end{proposition}

\begin{proof}
\begin{enumerate}
\item For a decomposable chain from Proposition \ref{ItoRepMCPropos}~we can
write (with $q_{1}=p_{1}$ and $p_{1}+p_{2}=q$) $P=%
\sum_{i=1}^{p_{1}}(E_{i}^{r}\otimes G_{i}^{n}~\lambda
_{i})+\sum_{j=p_{1}+1}^{q}(G_{j}^{r}\lambda _{j}\otimes E_{j}^{n})$ .~The
result follows from the fact that $\sum_{i=1}^{m}G_{i}^{d}~\lambda _{i}$ $%
\in \widehat{P}_{d}$ for any integers $m$ and $d,$ and by shifting the
summation index in the second sum.

\item Follows from Proposition \ref{ItoRepMCPropos}(2) by setting $%
C=\sum_{j=p_{1}+1}^{q}G_{j}^{r}\lambda _{j}$ noting that $C$ $\in \widehat{%
P}_{r}$

\item Follows from Proposition \ref{ItoRepMCPropos}(1) as above.
\end{enumerate}
\end{proof}

The transition matrix representation (\ref{DiagnolizableMC}) above is not
unique to an uncoupled Markov process. In fact, any Markov process $y_{t}$ %
on joint state space $\{e_{i}\}_{i=1}^{rn}$ whose transition matrix $P$ can
be written in the form (\ref{DiagnolizableMC}) is said to be \textbf{%
diagonalizable. }We will shortly see some sufficient conditions for
diagonalizability in the context of MDPs. An important property of
diagonalizable Markov processes is that the \textit{marginal} probabilities
of the component processes evolve in accordance with stochastic matrices
given by the diagonal decomposition, and in fact this condition is also
sufficient to guarantee diagonalizability:

\begin{proposition}
\label{ProposCondnsDiag}Given a diagonalizable Markov process $%
y_{t}=z_{t}\otimes x_{t}$ whose transition matrix has the diagonal
representation (\ref{DiagnolizableMC}), the marginal probability
distributions $p_{z}$ and $p_{x}$ of the component processes $z_{t}$ and $%
x_{t}$ evolve in accordance with $\dot{p}_{z}(t)=Cp_{z}(t)~$ and $\dot{p}%
_{x}(t)=Ap_{x}(t)$ respectively.

Conversely, given a decomposable Markov process $y_{t}=z_{t}\otimes x_{t}$ %
such that the marginal probability distributions $p_{z}$ and $p_{x}$ of $%
z_{t}$ and $x_{t}$ evolve on $\{e_{i}\}_{i=1}^{r}$ and $\{e_{i}\}_{i=1}^{n}$
in accordance with $\dot{p}_{z}(t)=Cp_{z}(t)~$ and $\dot{p}_{x}(t)=Ap_{x}(t)~
$respectively, where $A\in \widehat{P}_{n}$ and $C\in \widehat{P}_{r},$
then $y_{t}$ is diagonalizable with the~representation given by (\ref%
{DiagnolizableMC}).
\end{proposition}

From Propositions \ref{ProposTransRep} and \ref{ItoRepMCPropos} we get the
following:

\begin{proposition}
\label{ProposDecomposIto}Let $y_{t}=z_{t}\otimes x_{t}$ be a Markov process
in $r\times n$ states where $z_{t}\in \{e_{i}\}_{i=1}^{r}$ and $x_{t}\in
\{e_{i}\}_{i=1}^{n}.$ Then sample paths of $y_{t}$ can be written as 
\begin{equation*}
dy_{t}=(z_{t}\otimes dx_{t})+(dz_{t}\otimes x_{t})+(dz_{t}\otimes dx_{t})
\end{equation*}%
If $y_{t}$ is \emph{decomposable}, then the sample paths can be decomposed
into 
\begin{eqnarray*}
z_{t}\otimes dx_{t} &=&z_{t}\otimes \sum_{j=1}^{m}G_{j}(z)x_{t}dN_{j}(z_{t})
\\
dz_{t}\otimes x_{t} &=&\sum_{i=1}^{s}H_{i}(z)z_{t}dM_{i}(z_{t})\otimes x_{t}
\\
dz_{t}\otimes dx_{t} &=&0
\end{eqnarray*}%
where $G_{j}(z)\in \mathbb{G}^{n},H_{i}\in \mathbb{G}^{r}$ for each $z,x$
and $N_{j},M_{i}$ are doubly stochastic (Markov modulated) Poisson counters.
Furthermore, if $y_{t}$ is a Cascade MC then we get the following decoupled
Ito representation%
\begin{eqnarray*}
dz &=&\sum_{i=1}^{s}H_{i}zdM_{i} \\
dx &=&\sum_{i=1}^{m}G_{i}(z)xdN_{i}(z)
\end{eqnarray*}
\end{proposition}

\begin{remark}
If $y_{t}$ is non-decomposable, the term $dz_{t}\otimes dx_{t}$ is non-zero,
so we can not write sample paths in decoupled form.
\end{remark}

\section{Diagonalizable Markov Decision Processes}\label{SectionDiagMP} 

\subsubsection{Properties of Diagonalizable MDPs}

If the MDP is diagonalizable, then some simplifications of the solutions
presented above are possible. Once again consider optimal control problem (%
\ref{OptCtrlProblem}) except that now the cascade is diagonalizable. Using
notation of Section \ref{SectionMaxPrinciple}, the joint probabilities $%
p_{i}$ satisfy (assume stationarity of $z(t))$ $\qquad \qquad \qquad $%
\begin{equation*}
\dot{p}_{i}=(A_{i}+\sum_{j=1}^{p}B_{ij}D_{ij})p_{i}
\end{equation*}

From Proposition \ref{ProposCondnsDiag} and the fact that the marginal
probability vector of $x(t)$ is $\sum_{i=1}^{r}p_{i}$ we must have, for some
stochastic matrix $\overline{A},$ 
\begin{equation}
\sum_{i=1}^{r}(A_{i}+\sum_{j=1}^{p}B_{ij}D_{ij})p_{i}=\overline{A}%
\sum_{i=1}^{r}p_{i}  \label{FullyDecplNecCondn}
\end{equation}

Thus we have the following useful lemma:

\begin{lemma}
\label{LemmaDiagonalizability}Let $z\in \{e_{i}\}_{i=1}^{r}$ $,x\in
\{e_{i}\}_{i=1}^{n}$ \ and $A(t,z),B_{j}(t,z),u_{j}(t,z,x),$ $j=1..p,$ and a
cascade MDP on $z\otimes x$ be as defined in Section \ref%
{SectionCascadeMDPModel1}. As before, use shorthand $A_{i}\equiv
A(t,e_{i}),~B_{ij}\equiv B_{j}(t,e_{i})$, and $u_{j}(t,e_{i},x)$ as the
diagonal matrix $D_{ij}$.Then the resulting Markov process is diagonalizable
if and only if there exists a stochastic matrix $\overline{A}(t)$ such that
the joint probabilities written as vectors $%
\{p_{i}(t)=[p_{i1},p_{i2}...p_{in}]^{T},i=1..r\}$ where $p_{ik}(t)=$ $\Pr
\{z(t)=e_{i},x(t)=e_{k}\}$ at each $t$ satisfy the equation (\ref%
{FullyDecplNecCondn}), assuming that $z(t)$ is stationary\footnote{%
Similar equation can be derived for non-stationary $z(t)$ but not needed in
this paper}
\end{lemma}

\begin{corollary}
\label{CorrSufficientDiagz}(\textbf{Sufficient Conditions for diagonalizable
MDP}).The cascade MDP\ defined in the hypothesis of Lemma \ref%
{LemmaDiagonalizability} is diagonalizable if any of following hold:

\begin{enumerate}
\item $A(t,z),~B_{j}(t,z)$ and $u_{j}(t,z,x)$ are$\ $independent of $z$, $%
j=1,2..p.$ That is, $\mathcal{U}$ is restricted to the set of measurable
functions on the space $\mathbb{R}^{+}\mathbb{\times }\{e_{i}\}_{i=1}^{n}$ %
only (i.e. no feedback allowed on state $z$ )

\item For each $x\in \{e_{k}\}_{k=1}^{n}$ and $t,$ the sum $%
A(t,z)+\sum_{j=1}^{p}u_{j}(t,z,x)B_{j}(t,z)$ is independent of $z$ for all
admissible controls $u_{j}.$

\item For each $i,k$ such that the $k^{\prime }th$ row of $%
A_{i}+\sum_{j=1}^{p}B_{ij}D_{ij}$ does not vanish for all $t$ and admissible
controls $D_{ij}$,~the marginal probabilities $p_{i}^{Z}(t)\equiv \Pr
\{z(t)=e_{i}\}$ and $p_{k}^{X}(t)\equiv \Pr \{x(t)=e_{k}\}$ are
uncorrelated, i.e. $p_{ik}(t)=p_{i}^{Z}(t)p_{k}^{X}(t)$
\end{enumerate}
\end{corollary}

\begin{proof}
The first and second conditions are trivial. For the third, note that if $%
p_{ik}(t)=p_{i}^{Z}(t)p_{k}^{X}(t)$ then we can write the $m^{\prime }th$ \
row of the left hand side of (\ref{FullyDecplNecCondn}) as in fully expanded
form, using notation $(M)_{ij}$ for the $(i,j)^{th}$ entry of matrix $M$%
\begin{eqnarray*}
&&\sum_{k=1}^{n}\sum_{i=1}^{r}\sum_{j=1}^{p}(A_{i}+B_{ij}D_{ij})_{mk}p_{ik}
\\
&=&\sum_{k=1}^{n}\sum_{i=1}^{r}%
\sum_{j=1}^{p}(A_{i}+B_{ij}D_{ij})_{mk}p_{i}^{Z}p_{k}^{X} \\
&=&\sum_{k=1}^{n}(\sum_{i=1}^{r}p_{i}^{Z}%
\sum_{j=1}^{p}(A_{i}+B_{ij}D_{ij}))_{mk}\sum_{l=1}^{r}p_{lk}
\end{eqnarray*}

Setting $\overline{A}=$ $\sum_{i=1}^{r}p_{i}^{Z}%
\sum_{j=1}^{p}(A_{i}+B_{ij}D_{ij})$which is readily verified to be a
stochastic matrix, the result follows from Lemma \ref{LemmaDiagonalizability}%
.
\end{proof}

\subsubsection{Some Problems on Diagonalizable MDPs}\label{NoZFbck}

Note that from (\ref{FullyDecplNecCondn}) above, a diagonalizable MDP\ can
be rewritten as a partial feedback problem, by possibly introducing matrices 
$\overline{A}_{0}(t),\overline{B_{i}}(t)$ and controls $\overline{u_{j}}%
(t,x)$ such that $\overline{A}(t)=\overline{A}_{0}(t)+\sum_{j=1}^{\overline{p%
}}\overline{u_{j}}(t,x)\overline{B_{i}}(t).$ Thus all optimal control
problems on diagonalizable MDPs are in the category of partial feedback
problems.

Consider, once again the optimal control problem (\ref{OptCtrlProblem})
except that now the MDP is diagonalizable. Simplified solutions are
available in the following two cases.

\begin{theorem}
\label{ThmDiagCase}Let $z\in \{e_{i}\}_{i=1}^{r}$ $,x\in
\{e_{i}\}_{i=1}^{n}~ $and $A_{0},A$,$B_{i}$,$T$, $\mathcal{U},\psi $,$\Phi $,%
$L,\eta $,be as defined for the cascade MDP\ on $z\otimes x$  of Theorem \ref%
{TheoremBellman}. In addition, let $A$,$B_{i}$ and $\mathcal{U}$ satisfy the
hypothesis of Corollary \ref{CorrSufficientDiagz}.1. Then if the cost
functional $L$ or terminal condition $\Phi $ do not depend on $z,$ the to
the optimal control problem defined in (\ref{OptCtrlProblem}) has the
solution 
\begin{eqnarray*}
\eta ^{\ast } &=&\mathbb{E}k^{T}(0)x(0) \\
u^{\ast } &=&\arg \min_{u(x)\in \mathcal{U}}(\sum_{i=1}^{p}u_{i}k^{T}B_{i}x+%
\psi (u))
\end{eqnarray*}%
where $k$ satisfies the vector differential equation 
\begin{eqnarray*}
\dot{k} &=&-A^{T}k-L^{T}e_{1}-\min_{u(x)\in \mathcal{U}}(%
\sum_{i=1}^{p}u_{i}k^{T}B_{i}x+\psi (u)) \\
k(T) &=&\Phi ^{T}e_{1}
\end{eqnarray*}
\end{theorem}

\begin{proof}
In this case since we have no dependence of $A_{1,}B_{j}$ or $u_{j}$ on $z$ %
and neither that of $L$ or $\Phi $ the Bellman equation (\ref%
{BellmanPartiallyDec}) reduces to the single state Bellman equation (See
Theorem 1 in [2]) defined on the state space of $x(t).$ Hence we can use a
much simplified version of the Bellman equation to find the optimal control. 
\textit{Note, however, this does not necessarily imply complete
independence, }in the sense that the marginal probabilities may still be
correlated.
\end{proof}

\begin{remark}
Note that in view of the introductory remark in Section \ref{NoZFbck} the
condition requiring satisfaction of hypothesis of Corollary \ref%
{CorrSufficientDiagz}.1 is not necessary for a diagonalizable MDP.
\end{remark}

\begin{theorem}
Let $z\in \{e_{i}\}_{i=1}^{r}$ $,x\in \{e_{i}\}_{i=1}^{n}$ and $A_{0},A$,$%
B_{i}$,$T$,$\psi $,$\Phi $,$L,\eta $,be as defined for the cascade MDP\ on $%
z\otimes x$  of Theorem \ref{TheoremBellman}. \ Let $\mathcal{U~}$be
restricted to the set of measurable functions on the space $\mathbb{R}^{+}%
\mathbb{\times }\{e_{i}\}_{i=1}^{n}$(i.e. no feedback allowed on state $z$
), and further let the MDP satisfy the hypothesis of Corollary \ref%
{CorrSufficientDiagz}.3. Using notation $c_{i}(t)=\Pr
\{z(t)=e_{i}\},A_{i}(t)=A(t,e_{i}),B_{ij}(t)=$ $B_{j}(t,e_{i})$ the optimal
control problem defined in (\ref{OptCtrlProblem}) has the solution 
\begin{eqnarray*}
\eta ^{\ast } &=&c^{T}(0)\mathbb{E}k^{T}(0)x(0) \\
u^{\ast } &=&\arg \min_{u(x)\in \mathcal{U}}(\sum_{i=1}^{p}u_{i}k^{T}(%
\sum_{i=1}^{r}c_{i}B_{ij})x+\psi (u))
\end{eqnarray*}%
where $k$ satisfies the vector differential equation 
\begin{eqnarray*}
\dot{k}&=&-A_{0}^{T}k-(\sum_{i=1}^{r}c_{i}A_{i}^{T})k-L^{T}c \\
&&-\min_{u(x)\in \mathcal{U}}(\sum_{i=1}^{p}u_{i}k^{T}(%
\sum_{i=1}^{r}c_{i}B_{ij})x+\psi (u)) \\
k(T)&=&\Phi ^{T}c
\end{eqnarray*}
\end{theorem}

\begin{proof}
In this case, if we examine the Hamiltonian in (\ref{Hamiltonian}) we note
that in the term to be minimized becomes $\sum_{i=1}^{r}\sum_{j=1}^{p}%
\sum_{k=1}^{n}p_{ik}(u_{jk}q_{i}^{T}B_{ij}),$ (assuming no control cost) But
since $p_{ik}=p_{k}c_{i}$ where $p_{k}$ and $c_{i}$  are the marginal
probabilities of $x(t)=e_{k}$ and $z(t)=e_{i}$ respectively, this otherwise
non trivial minimization becomes trivial since we can now interchange the
summation order to write this sum as by writing $B_{j}=.%
\sum_{i=1}^{r}c_{i}B_{ij}$  $\sum_{k=1}^{n}p_{k}\sum_{j=1}^{p}u_{jk}(%
\sum_{i=1}^{r}c_{i}(q_{i}^{T}B_{j}))$ and since $p_{k}\geq 0$ we achieve
minimization by choosing $u_{jk}$ to minimize $%
(\sum_{i=1}^{r}c_{i}(q_{i}^{T}B_{j})).$ This then becomes the condition for
the minimum in the costate equation as well, and hence we have removed
dependency of the costate equation on the state $p$ and so we can solve the
costate equation (i.e. this becomes the single state Bellman equation).
\end{proof}

\section{Portfolio Optimization}

\subsection{\label{SecPOBg}Background: Portfolio Value, Wealth and Investment%
}

A portfolio consists of a finite set~of assets (such as stocks or bonds),
with the \emph{weight process} $x_{t}$ denoting the vector of amounts (also
called allocations or weights) of the assets. The \emph{price process} $%
z_{t}$ denotes the vector of market prices of the assets We define the
portfolio value $v(t,z,x)$ as the net value of the current asset holdings
for weights $x$ and prices $z$.~If $x(t)$ and $z(t)$ take values in finite
sets of standard basis vectors, then $v$ can be represented by the matrix $%
V(t)$ as $v(t,z,x)=z^{T}V(t)x$.~Using the Ito rule, we can write 
\begin{equation*}
dv=dz^{T}V^{T}x+z^{T}V^{T}dx
\end{equation*}%
In a non self-financing model, depending on the current value of the
portfolio, a weight shift will require buying/selling assets using an
investment (or consumption, which is the negative of the investment). If $%
s(t)$ represents the net \textbf{investment }into the portfolio up to time $t$%
, the incremental investment is the change in the portfolio value due to
weight shift. Hence,%
\begin{equation}
ds=z^{T}V^{T}dx  \label{InvestDef}
\end{equation}%
Similarly, the \textbf{wealth} of the portfolio (i.e. its intrinsic worth)
at time $t$ is defined as $w(t)=v(t)-s(t).$ So that the wealth represents
the net effect of changes in asset prices, and we can write%
\begin{equation}
dw=dz^{T}V^{T}x  \label{WealthDef}
\end{equation}

\subsection{\protect\bigskip \label{AssetExSF43}Self-Financing Portfolio
Problem}

We assume there are two stocks $S_{1}$ and $S_{2}$ whose prices each evolve
independently on a state space of $\{-1,1\}.$ Assume a portfolio that can
shift weights between the two assets with allowable weights $W$ of $%
(2,0),(1,1),(0,2)$ so that the portfolio has a constant total position (of $2
$). Further, we allow only weight adjustments of $+1$ or $-1$ for each
asset, and we further restrict the weight shifts to only those that do not
cause a change in net value for any given asset price. The latter condition
makes the portfolio self-financing.

The resulting process can be modeled as a cascade MDP. Let $z_{t}$ be the
(joint) prices of the two assets with prices $(-1,-1)$, $(-1,1)$, $(1,-1)$, $%
(1,1)$ represented as states $e_{1},e_{2},e_{3},e_{4}$ respectively. Let $%
x_{t}$ be the choice of weights with weights $(0,2)$, $(1,1)$, $(2,0)$
represented as states $e_{1},e_{2},e_{3}$ respectively. Transition rates of $%
z_{t}$ are determined by some pricing model, whereas the rates of $x_{t}$
which represent allowable weight shifts are controlled by the portfolio
manager. The portfolio value $v(z_{t},x_{t})$ can be written using its
matrix representation, $v(z,x)=z^{T}Vx$, where $V$ is

\begin{equation}
{\tiny V=%
\begin{pmatrix}
-2 & 2 & 0 & 2 \\ 
-2 & 2 & -2 & 2 \\ 
-2 & 0 & -2 & 2%
\end{pmatrix}%
}  \label{app:AssetVMSFP}
\end{equation}%
The portfolio manager is able to adjust the rate of increasing the first
weight by an amount $u$ and, independently that of decreasing the first
weight by an amount $d$ (which has the effect of simultaneously decreasing
or increasing the weight of the second asset). The resulting transitions of $%
x_{t}$ depend on $z_{t}$ (see Figure \ref{fig:appSFPModel} ) and transition matrices $P(z)$ of
the weights $x_{t}$ can be written as $P(z)=$\thinspace $A(z)+uB(z)+dD(z)$,
where $A(z),B(z),D(z)$ are: {\tiny {\ 
\begin{equation*}
\begin{tabular}{cccc}
$A(e_{1})=\frac{1}{2}%
\begin{pmatrix}
-1 & 1 & 0 \\ 
1 & -2 & 1 \\ 
0 & 1 & -1%
\end{pmatrix}%
$ & $A(e_{2})=\frac{1}{2}%
\begin{pmatrix}
-1 & 1 & 0 \\ 
1 & -1 & 0 \\ 
0 & 0 & 0%
\end{pmatrix}%
$ & $A(e_{3})=\frac{1}{2}%
\begin{pmatrix}
0 & 0 & 0 \\ 
0 & -1 & 1 \\ 
0 & 1 & -1%
\end{pmatrix}%
$ & $A(e_{4})=\frac{1}{2}%
\begin{pmatrix}
-1 & 1 & 0 \\ 
1 & -2 & 1 \\ 
0 & 1 & -1%
\end{pmatrix}%
$ \\ 
$B(e_{1})=%
\begin{pmatrix}
-1 & 0 & 0 \\ 
1 & -1 & 0 \\ 
0 & 1 & 0%
\end{pmatrix}%
$ & $B(e_{2})=%
\begin{pmatrix}
-1 & 0 & 0 \\ 
1 & 0 & 0 \\ 
0 & 0 & 0%
\end{pmatrix}%
$ & $B(e_{3})=%
\begin{pmatrix}
0 & 0 & 0 \\ 
0 & -1 & 0 \\ 
0 & 1 & 0%
\end{pmatrix}%
$ & $B(e_{4})=%
\begin{pmatrix}
0 & 0 & 0 \\ 
0 & -1 & 0 \\ 
0 & 1 & 0%
\end{pmatrix}%
$ \\ 
$D(e_{1})=%
\begin{pmatrix}
0 & 1 & 0 \\ 
0 & -1 & 1 \\ 
0 & 0 & -1%
\end{pmatrix}%
$ & $D(e_{2})=%
\begin{pmatrix}
0 & 1 & 0 \\ 
0 & -1 & 0 \\ 
0 & 0 & 0%
\end{pmatrix}%
$ & $D(e_{3})=%
\begin{pmatrix}
0 & 0 & 0 \\ 
0 & 0 & 1 \\ 
0 & 0 & -1%
\end{pmatrix}%
$ & $D(e_{4})=%
\begin{pmatrix}
0 & 1 & 0 \\ 
0 & -1 & 1 \\ 
0 & 0 & -1%
\end{pmatrix}%
$%
\end{tabular}%
\end{equation*}%
}} For $P(z)$ to be a proper transition matrix we require admissible
controls $u,d$ to satisfy $\left\vert u\right\vert ,\left\vert d\right\vert
\leq \frac{1}{2}$.~The portfolio manager may choose $u,d$ in accordance with
current values of $x_{t}$ and $z_{t}$ so that $u,d$ are Markovian feedback
controls $u(t,z_{t},x_{t})$ and $d(t,z_{t},x_{t}).$ Note that this model
differs from the traditional Merton-like models where only feedback on the
total value $v_{t}$ is allowed. Note that it is the self-financing
constraint that leads to the dependence on the current price $z_{t}$ of the
transitions of $x$ which allows us to model this problem as a cascade.

\begin{center}
\begin{figure}[th]
\centering
\subfigure[$z=e_1$ or $z=e_4$]{
\begin{tikzpicture}[->,>=stealth',shorten >=1pt,auto,node distance=2.4cm,
semithick]
\tikzstyle{every state}=[fill=white,draw=black,thick,text=black,scale=1]
\node[state]         (A)              {\small{$(1,1)$}};
\node[state]         (B) [left of=A] {\small{$(0,2)$}};
\node[state]         (C) [right of=A] at (A) {\small{$(2,0)$}};
\path (A) edge  [bend right] node[above] {$\frac{1}{2}+d$} (B);
\path (A) edge  [bend right] node[below] {$\frac{1}{2}+u$} (C);
\path (B) edge  [bend right] node[below] {$\frac{1}{2}+u$} (A);
\path (C) edge  [bend right] node[above] {$\frac{1}{2}+d$} (A);
\end{tikzpicture}
\label{fig:appSFP1}
} 
\subfigure[$z=e_2$]{
\begin{tikzpicture}[->,>=stealth',shorten >=1pt,auto,node distance=2.4cm,
semithick]
\tikzstyle{every state}=[fill=white,draw=black,thick,text=black,scale=1]
\node[state]         (A)              {\small{$(1,1)$}};
\node[state]         (B) [left of=A] {\small{$(0,2)$}};
\path (A) edge  [bend right] node[above] {$\frac{1}{2}+d$} (B);
\path (B) edge  [bend right] node[below] {$\frac{1}{2}+u$} (A);
\end{tikzpicture}
\label{fig:appSFP2}
} 
\subfigure[$z=e_3$]{
\begin{tikzpicture}[->,>=stealth',shorten >=1pt,auto,node distance=2.4cm,
semithick]
\tikzstyle{every state}=[fill=white,draw=black,thick,text=black,scale=1]
\node[state]         (A)              {\small{$(1,1)$}};
\node[state]         (C) [right of=A] at (A) {\small{$(2,0)$}};
\path (A) edge  [bend right] node[below] {$\frac{1}{2}+u$} (C);
\path (C) edge  [bend right] node[above] {$\frac{1}{2}+d$} (A);
\end{tikzpicture}
\label{fig:appSFP3}
} \label{fig:appSFPModel}
\caption[Self Financing Portfolio Transition Diagrams]{Transition diagram of
weight $x(t)$ in the self-financing portfolio for various asset prices $z(t)$
are shown in \subref{fig:appSFP1}, \subref{fig:appSFP2} and \subref{fig:appSFP3}.
States $e_{1}$,$e_{2}$,$e_{3}$,$e_{4}$ of $z(t)$ correspond to price vectors
(-1,-1),(-1,1),(1,-1),(1,1) respectively. Self-transitions are omitted for
clarity.}
\end{figure}
\end{center}

Consider the problem of maximizing the expected terminal value $v(T)$ of the
portfolio for a fixed horizon $T$ for the above self-financing portfolio
model~\ref{SecSFPModel}. With $x,z,u,d,V,A,B,D$ as defined thereof, we wish
to maximize the performance measure given by 
\begin{equation*}
\eta (u,d)=\mathbb{E(}v(T))
\end{equation*}%
Using Theorem \ref{TheoremBellman} we see the solution to this \textbf{OCP-I}
problem is obtained by solving the matrix equation with boundary condition $%
K(T)=-V$ 
\begin{equation}
\dot{K}=-KC-A^{T}(z)K+\frac{1}{2}\left\vert K^{T}B(z)\right\vert +\frac{1}{2}%
\left\vert K^{T}D(z)\right\vert   \label{app:SelfFnSoln}
\end{equation}%
with the optimum performance measure and controls (in feedback form) given
by 
\begin{eqnarray*}
\eta ^{\ast } &=&z^{T}(0)K^{T}(0)x(0) \\
u^{\ast }(t,z,x) &=&-\frac{1}{2}\sgn(z^{T}K(t)^{T}B(z)x) \\
d^{\ast }(t,z,x) &=&-\frac{1}{2}\sgn(z^{T}K(t)^{T}D(z)x)
\end{eqnarray*}%
with $K(t)$ being the solution to (\ref{app:SelfFnSoln}). Some solutions for (%
\ref{app:SelfFnSoln}) and corresponding optimal controls are plotted for $T=1,15~
$in Figure \ref{fig:FigOptimalStrategy_SF_43} for
various initial conditions (mixes of the assets in the portfolio initially).
Results also show that as $T\rightarrow \infty ,$the value of $\eta ^{\ast }$
approaches a constant value of $0.4725$ regardless of the initial values $%
z(0),x(0).$ That is the maximal possible terminal value for the portfolio is 
$0.4725.$ However, we do not see a steady state constant value for the
optimal controls $u^{\ast }(z,x)$ and $d^{\ast }(z,x)$ and that
near the portfolio expiration date, more vigorous buying/selling activity is
necessary. \ If the matrix $C$ were reducible or time-varying in our
example, multiple steady-states are possible as $T\rightarrow \infty $ and
the initial trading activity will be more significant.

\begin{figure}[!ht]
\centering
\includegraphics[width=\linewidth]{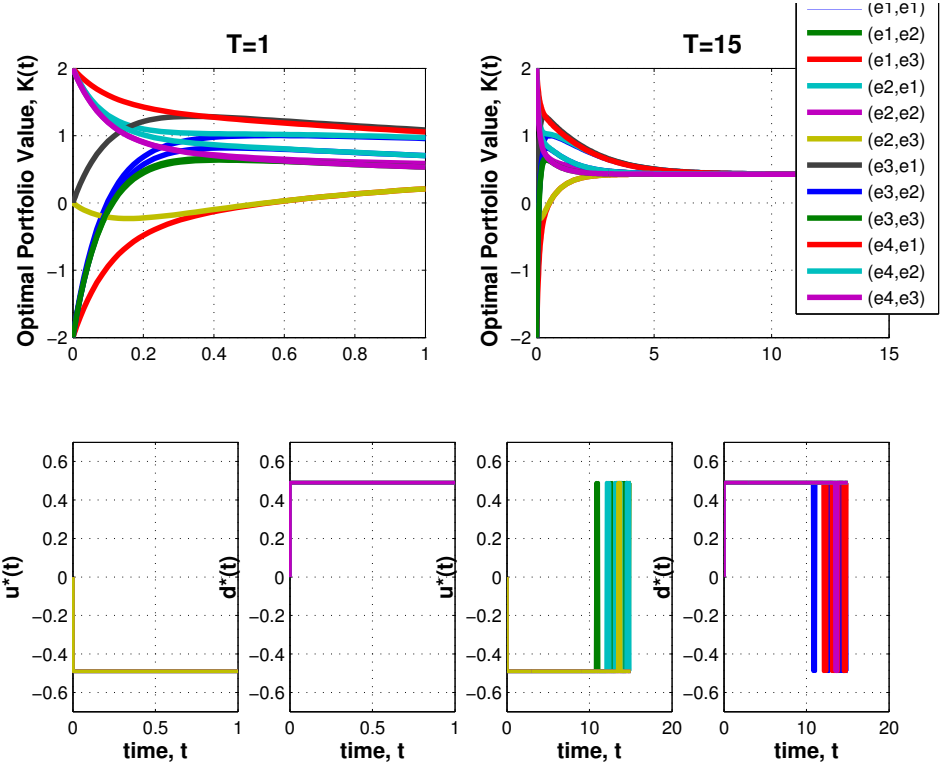}
\caption[Self-Financing Portfolio]{Solution to problem \ref{AssetExSF43}.
Minimum Return Function $k(t,z,x)$,
optimal up controls $u^{\ast }(t,z,x)$ and down controls $d^{\ast }(t,z,x)$
for the self-financing portfolio with maximal terminal wealth. Figures(a) and
(b) are for $T=1$ and $T=15$ respectively. Various $(z,x)$ values are
represented by the vectors $(e_{i},e_{j})$.}
\label{fig:FigOptimalStrategy_SF_43}
\end{figure}

\subsection{An Investment-Consumption Portfolio Problem}\label{AssetExIC}

An alternate model for portfolio allocation than discussed in the
self-financing Portfolio example (Section ) is presented as a \textbf{OCP-I }%
problem in this section. If we do not restrict the weight adjustments in the
model of Section \ref{SecSFPModel} to cases which keep the value a constant,
(i.e. we allow only weight adjustments of $+1$ or $-1$ for each asset,
regardless of the current portfolio value) we get a non self-financing
portfolio. The difference in the portfolio value as a result of weight shift
must be the result of an equivalent investment or consumption. Once again,
modeling this is as a cascade with $z_{t}$ and $x_{t}$ as in Section \ref%
{SecSFPModel}, the portfolio value matrix (\ref{app:AssetVMSFP}) is replaced by 
\begin{equation}
V=%
\begin{pmatrix}
-2 & 2 & -2 & 2 \\ 
-2 & 0 & 0 & 2 \\ 
-2 & 2 & 2 & 2%
\end{pmatrix}
\label{Asset1DVM}
\end{equation}%
As before, the portfolio manager can control the up and down rates $u,d$ %
resulting in the transitions of $x_{t}\ $(See Figure) described by the
matrices $P(z)$ $=$\thinspace $A(z)+uB(z)+dD(z)$ with%
\begin{equation*}
\begin{tabular}{cc}
$A(z)=%
\begin{pmatrix}
-\frac{1}{2} & \frac{1}{2} & 0 \\ 
\frac{1}{2} & -1 & \frac{1}{2} \\ 
0 & \frac{1}{2} & -\frac{1}{2}%
\end{pmatrix}%
$ &  \\ 
$B(z)=%
\begin{pmatrix}
-1 & 0 & 0 \\ 
1 & -1 & 0 \\ 
0 & 1 & 0%
\end{pmatrix}%
$ & $D(z)=%
\begin{pmatrix}
0 & 1 & 0 \\ 
0 & -1 & 1 \\ 
0 & 0 & -1%
\end{pmatrix}%
$%
\end{tabular}%
\end{equation*}%
and admissibility condition $\left\vert u\right\vert ,\left\vert
d\right\vert \leq \frac{1}{2}$. Note in this cascade model the matrices $%
A,B,D$ do not depend on $z$ but we will see next that the performance
measure does depend on $z.$

\begin{center}
\begin{figure}[tbp]
\centering
\begin{tikzpicture}[->,>=stealth',shorten >=1pt,auto,node distance=2.4cm,
semithick]
\tikzstyle{every state}=[fill=white,draw=black,thick,text=black,scale=1]
\node[state]         (A)              {\small{$(1,1)$}};
\node[state]         (B) [left of=A] {\small{$(0,2)$}};
\node[state]         (C) [right of=A] at (A) {\small{$(2,0)$}};
\path (A) edge  [bend right] node[above] {$\frac{1}{2}+d$} (B);
\path (A) edge  [bend right] node[below] {$\frac{1}{2}+u$} (C);
\path (B) edge  [bend right] node[below] {$\frac{1}{2}+u$} (A);
\path (C) edge  [bend right] node[above] {$\frac{1}{2}+d$} (A);

\end{tikzpicture}
\label{fig:AssetEx1}
\caption[Investment/Consumption Portfolio Transition Diagram]{Transition
diagram of weights $x(t)$ for controls $u$,$d$ in the investment/consumption
portfolio. Self-transitions are omitted for clarity.}
\end{figure}
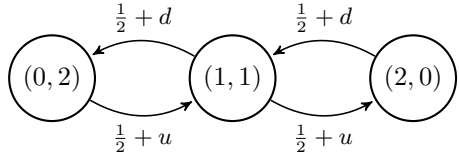
\end{center}

\subsubsection{\label{Asset_Example_1_1}Problem 1: Minimal Investment}

We wish to minimize the total amount of investment up to a fixed horizon $T.$ %
We can write a performance measure $\eta (u,d)$ that represents the net
investment into the portfolio up to time $T$ as 
\begin{equation*}
\eta (u,d)=\mathbb{E(}s(T))
\end{equation*}%
where $s(t)$ is the \emph{investment} process\ (See Appendix \ref{SecPOBg}).
Using (\ref{InvestDef}) 
\begin{equation*}
\mathbb{E(}ds(t))=\mathbb{E(}z^{T}V^{T}dx)=\mathbb{E(}z^{T}V^{T}(A+uB+dD)x)dt
\end{equation*}%
Writing the matrix $\Phi (u,d)=V^{T}(A+uB+dD)$ 
\begin{equation*}
\eta (u,d)=\mathbb{E}\int_{0}^{T}z^{T}(t)\Phi (u,d)x(t)dt
\end{equation*}%
The goal then is to choose $u,d$ so as to minimize $\eta (u,d)$ subject to $%
u,d\in \mathcal{U~}$where the admissibility set $\mathcal{U}$ is the set of
past measurable functions $u(z,x)$ such that $\left\vert u(z,x)\right\vert
\leq \frac{1}{2}$ for each $z,x.$ Using Theorem \ref{TheoremBellman} we see
the solution to this \textbf{OCP-I} problem is obtained by solving the
matrix equation with boundary condition $K(T)=0$ 
\begin{equation}
\dot{K}=-KC-A^{T}(K+V)+\frac{1}{2}\left\vert B^{T}(K+V)\right\vert +\frac{1}{%
2}\left\vert D^{T}(K+V)\right\vert  \label{Asset1_1_Bellman}
\end{equation}%
(where the notation $\left\vert M\right\vert $ for a matrix $M$ above
represents the element-by-element absolute value of a matrix) with the
optimal performance measure and controls (in feedback form) given by 
\begin{eqnarray*}
\eta ^{\ast } &=&z^{T}(0)K(0)x(0) \\
u^{\ast }(t,z,x) &=&-\frac{1}{2}\sgn(z^{T}(K(t)+V)^{T}Bx) \\
d^{\ast }(t,z,x) &=&-\frac{1}{2}\sgn(z^{T}(K(t)+V)^{T}Dx)
\end{eqnarray*}%
where $K(t)$ is the solution to (\ref{Asset1_1_Bellman}). Some solutions to (%
\ref{Asset1_1_Bellman}) and corresponding optimal controls are plotted for $%
T=1,10$ in Figure \ref{fig:FigAssetEx_MinInvest}(a) and \ref%
{fig:FigAssetEx_MinInvest}(b). Results also show that as $T\rightarrow
\infty ,$the value of $\eta ^{\ast }/T$ approaches a constant value of $%
-0.535$ regardless of the initial values $z(0),x(0)$ and in this case we see
that the optimal controls $u^{\ast }(z,x)$ and $d^{\ast }(z,x)$ expressed in
matrix form ($u^{\ast }(z,x)$ written as $z^{T}u^{\ast }x$ etc.) 
\begin{equation*}
u^{\ast }=%
\begin{pmatrix}
\frac{1}{2} & \frac{1}{2} & -\frac{1}{2} & \frac{1}{2} \\ 
-\frac{1}{2} & \frac{1}{2} & -\frac{1}{2} & \frac{1}{2} \\ 
0 & 0 & 0 & 0%
\end{pmatrix}%
;~d^{\ast }=%
\begin{pmatrix}
0 & 0 & 0 & 0 \\ 
-\frac{1}{2} & -\frac{1}{2} & \frac{1}{2} & -\frac{1}{2} \\ 
\frac{1}{2} & -\frac{1}{2} & \frac{1}{2} & -\frac{1}{2}%
\end{pmatrix}%
\end{equation*}%
(the values of $u^{\ast }(z,e_{3})$ and $d^{\ast }(z,e_{1})$ are immaterial
as they do not impact the dynamics). This means that one can expect a
constant cash flow of $0.535$ by the above strategy, and that this value is
maximal. Note also that the optimal controls \emph{do} depend on $z$ and so
the resulting weight and asset probabilities are not independent.

\begin{figure}[!ht]
\centering
\includegraphics[width=\linewidth]{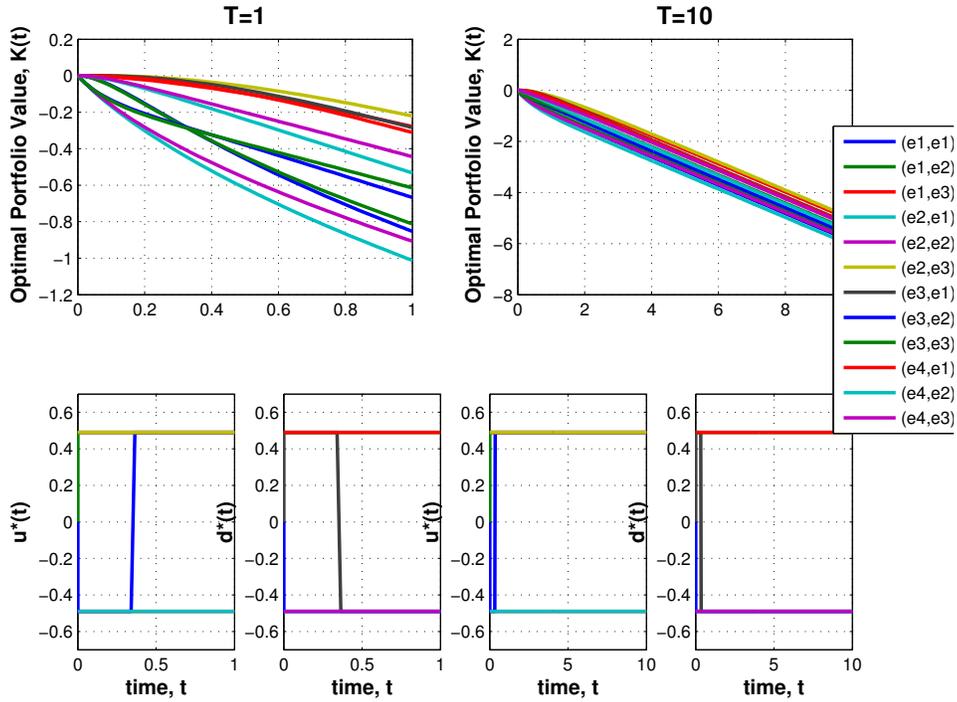}
\caption[Minimum Investment Portfolio]{Solution to problem \ref{Asset_Example_1_1}.
Minimum Return Function $k(t,z,x)$,
optimal up controls $u^{\ast }(t,z,x)$ and down controls $d^{\ast }(t,z,x)$
for the self-financing portfolio with maximal terminal wealth. Figures(a) and
(b) are for $T=1$ and $T=10$ respectively. Various $(z,x)$ values are
represented by the vectors $(e_{i},e_{j})$.}
\label{fig:FigAssetEx_MinInvest}
\end{figure}

\subsubsection{\label{Asset_Example_1_2}Problem 2 : Maximal Terminal Wealth}

In this case the performance measure that needs to be maximized is given by%
\begin{equation*}
\eta (u,d)=\mathbb{E(}w(T))=\mathbb{E}\int_{0}^{T}z^{T}(t)C^{T}V^{T}x(t)dt
\end{equation*}%
where $w(t)$ is the \emph{wealth }process (Appendix \ref{SecPOBg})~for $%
u,d\in \mathcal{U}$ as above. Again, from Theorem \ref{TheoremBellman} the
solution to this \textbf{OCP-I} problem is obtained by solving the matrix
equation with boundary condition $K(T)=0$ 
\begin{equation}
\dot{K}=-(K-V)C-A^{T}K+\frac{1}{2}\left\vert B^{T}K\right\vert +\frac{1}{2}%
\left\vert D^{T}K\right\vert ~  \label{Asset1_2_Bellman}
\end{equation}%
whose solution $K(t)$ gives the optimal performance measure and controls as:%
\begin{eqnarray*}
\eta ^{\ast } &=&z^{T}(0)K(0)x(0) \\
u^{\ast }(t,z,x) &=&-\frac{1}{2}\sgn(z^{T}K^{T}(t)Bx) \\
d^{\ast }(t,z,x) &=&-\frac{1}{2}\sgn(z^{T}K^{T}(t)Dx)
\end{eqnarray*}%
Some numerical results for the above problem with $V$ as in (\ref{Asset1DVM}%
) are plotted for $T=1,10$ in Figure \ref{fig:FigAssetEx_MaxWealth} (a) and %
\ref{fig:FigAssetEx_MaxWealth} (b).\ Results also show that as $T\rightarrow
\infty ,$the value of $\eta ^{\ast }/T$ approaches a constant value of $%
-0.533$ regardless of the initial values $z(0),x(0)$ and in this case we see
that the optimal controls $u^{\ast }(z,x)$ and $d^{\ast }(z,x)$ expressed in
matrix form ($u^{\ast }(z,x)$ written as $z^{T}u^{\ast }x$ etc.) are 
\begin{equation*}
u^{\ast }=%
\begin{pmatrix}
\frac{1}{2} & -\frac{1}{2} & \frac{1}{2} & \frac{1}{2} \\ 
-\frac{1}{2} & -\frac{1}{2} & \frac{1}{2} & -\frac{1}{2} \\ 
0 & 0 & 0 & 0%
\end{pmatrix}%
;~d^{\ast }=%
\begin{pmatrix}
0 & 0 & 0 & 0 \\ 
-\frac{1}{2} & \frac{1}{2} & -\frac{1}{2} & \frac{1}{2} \\ 
\frac{1}{2} & -\frac{1}{2} & -\frac{1}{2} & -\frac{1}{2}%
\end{pmatrix}%
\end{equation*}

\begin{figure}[!ht]
\centering
\includegraphics[width=\linewidth]{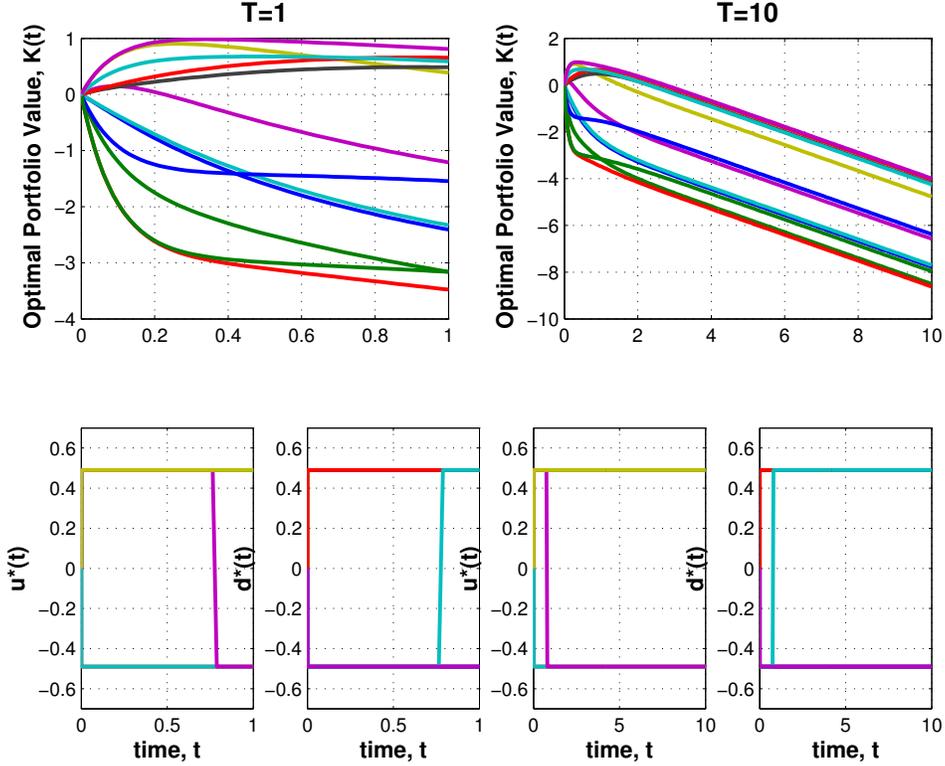}
\caption[Maximal Wealth Investment Portfolio]{Solution to problem \ref{Asset_Example_1_2}.
Minimum Return Function $k(t,z,x)$,
optimal up controls $u^{\ast }(t,z,x)$ and down controls $d^{\ast }(t,z,x)$
for the self-financing portfolio with maximal terminal wealth. Figures(a) and
(b) are for $T=1$ and $T=10$ respectively. Various $(z,x)$ values are
represented by the vectors $(e_{i},e_{j})$.}
\label{fig:FigAssetEx_MaxWealth}
\end{figure}

\subsubsection{\label{Asset_Example_1_3}Problem 3 - Minimal Investment with Partial Feedback}

In the investment/consumption model, the control matrices $A,B,D$ do not
depend on $z.$ As a result one may be tempted to think that a partial
feedback optimization problem, i.e. where the controls are allowed to depend
on $x$ but not $z$ would give the same optimal performance. However, one
sees from Theorem \ref{TheoremBellman} the solution to the minimal
investment case is obtained by solving the matrix equation subject to $%
K(T)=0 $%
\begin{eqnarray}
\dot{p}_{z} &=&Cp_{z};~p_{z}(0)=\mathbb{E}z(0)  \label{AssetPartialBllmn} \\
\dot{K} &=&-KC-A^{T}(K+V)+\frac{1}{2}\left\vert B^{T}(K+V)(\mathbf{e}%
_{r}p_{z}^{T})\right\vert  \notag \\
&&+\frac{1}{2}\left\vert D^{T}(K+V)(\mathbf{e}_{r}p_{z}^{T})\right\vert 
\notag
\end{eqnarray}%
where $p_{z}$ is the probability vector for $z$. And the optimal performance
and controls are given by 
\begin{eqnarray*}
\eta ^{\ast } &=&p_{z}^{T}(0)K^{T}(0)x(0) \\
u^{\ast }(t,x) &=&-\frac{1}{2}\sgn((\mathbf{e}%
_{r}p_{z}^{T}(t))(K(t)+V)^{T}Bx) \\
d^{\ast }(t,x) &=&-\frac{1}{2}((\mathbf{e}_{r}p_{z}^{T}(t))(K(t)+V)^{T}Dx)
\end{eqnarray*}%
where $K(t),p_{z}(t)$ are solutions to \ref{AssetPartialBllmn}. The best
performance in this case is worse than that in the full feedback case, as
indeed shown by numerical simulation as in \ref{fig:FigAssetEx_MinInvestNZ}%
(a),(b) for $T=1,10.$ Comparing with the respective minimum return functions
of the full feedback case, the steady state case maximal cash flow rate is
only $0.22$ compared to $0.533$.

\begin{figure}[!ht]
\centering
\includegraphics[width=\linewidth]{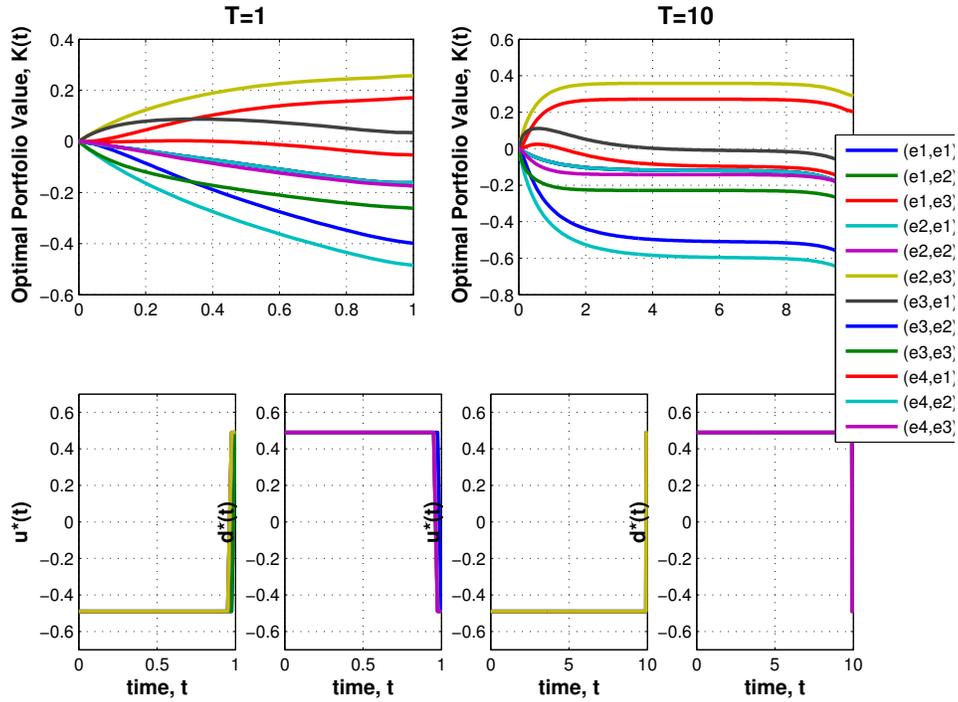}
\caption[Maximal Wealth Investment Portfolio With Partial Feedback]{Solution to problem \ref{Asset_Example_1_3}.
Minimum Return Function $k(t,z,x)$,
optimal up controls $u^{\ast }(t,z,x)$ and down controls $d^{\ast }(t,z,x)$
for the self-financing portfolio with maximal terminal wealth. Figures(a) and
(b) are for $T=1$ and $T=10$ respectively. Various $(z,x)$ values are
represented by the vectors $(e_{i},e_{j})$.}
\label{fig:FigAssetEx_MinInvestNZ}
\end{figure}

\subsection{Some Variations on Portfolio Problems}

Some variants of the examples presented here and in Section \ref{SecPOEx}
include the following:

\subsubsection{\textit{Utility Functions.and Discounting}}

In traditional portfolio optimization problems, one minimizes $\mathbb{E(}%
U(s(T)))$o maximizes $\mathbb{E(}U(w(T)))$ where $U(.)$ is a non-decreasing
and concave function, called the \textit{utility function}. In the above
examples, for simplicity of demonstration of the MDP techniques, we assumed $%
U(C)=C$. Utility functions are chosen based upon risk preferences of agents
and the financial environment, and some standard ones include the $U(C)=%
\frac{C^{\gamma }}{\gamma }$ (with $\gamma <1)$ or $U(C)=\log C.$ \
Furthermore, one may wish to optimize the \textit{discounted }value i.e $%
\mathbb{E}\int_{0}^{T}e^{-\alpha t}U(w(t))dt$ for some $\alpha >0$ instead.
The solutions to optimization problems of minimum investment and maximum
wealth in these cases are identical to (\ref{Asset1_1_Bellman}) and (\ref%
{Asset1_2_Bellman}) with $V$ replaced by $e^{-\alpha t}U(V).$

\subsubsection{\textit{Value Payoff Functions}}

The particular model we chose led to a value payoff $V$ as in (\ref%
{Asset1DVM}) though the problems presented above are completely generic with
respect to $V$ in that any other value of $V$ would work as well. In that
case we will have different mappings of the states $e_{1},e_{2},e_{3}$ of $x$
to the weights and that of $e_{1},e_{2},e_{3},e_{4}$ of $z$ to asset prices,
but it is only the value matrix $V$ that appears in any of the solutions and
these mappings are immaterial.

\subsubsection{\textit{Transaction Costs}}

If buying/selling of assets incurs a transaction cost then every weight
shift is associated with a cost. This can be modeled in terms of the control
costs. We can see that a value of $u=-\frac{1}{2}$ represents the case of a
minimal rate of buying the first asset, while $u=\frac{1}{2}$ represents a
maximal rate of buying the first asset. Likewise, the values $d=-\frac{1}{2}$
to $d=\frac{1}{2}$ represent the range of the rates of selling the first
asset. Hence a reasonable metric for the transaction costs would be $(u+%
\frac{1}{2})^{2}+(d+\frac{1}{2})^{2}.$ For example, a performance measure
like ($\alpha >0$)

\begin{equation*}
\eta (u,d)=\mathbb{E}\int_{0}^{T}\alpha ((u+\frac{1}{2})^{2}+(d+\frac{1}{2}%
)^{2})dt+\mathbb{E(}U(s(T)))
\end{equation*}

\section{Summary of Notations and Symbols}\label{NotSumm}

A stochastic basis $(\Omega ,\mathcal{F},\mathbb{F},\mathbb{P})$ is assumed
where $(\Omega ,\mathcal{F},\mathbb{P})$ is a probability space and $\mathbb{%
F}$ a filtration $(\mathcal{F}_{t})_{t\in T}$ on this space for a totally
ordered index set $T$ ($\subseteq \mathbb{R}^{+}$in our case). All
stochastic processes are assumed to be right continuous and adapted to $%
\mathbb{F}$.

\begin{tabular}{lp{6.53cm}}
$\mathbb{F}$ & A filtration $(\mathcal{F}_{t})_{t\in T}$ on $(\Omega ,%
\mathcal{F},\mathbb{P})$ where $T$ is a totally ordered index set \\ 
$\mathbb{G}^{n}$ & The space of square matrices of dimension $n$ of the form 
$F_{kl}-F_{ll}$ ~where $F_{ij}$ is the matrix of all zeros except for one in
the $i^{\prime }th$ row and $j^{\prime }th$ column \\ 
$\mathbb{E}^{n}$ & The space of diagonal $n\times n$ matrices with only $1$%
's or $0$'s \\ 
$I_{n}$ & $n\times n$ identity matrix, $I_{n}\in \mathbb{E}^{n}$ \\ 
$\mathbb{P}^{n}$ & The space of all stochastic matrices of dimension $n$ \\ 
$\{e_{i}\}_{i=1}^{n}$ & The set of $n$ standard basis vectors in $\mathbb{R}%
^{n}$ \\ 
$\mathbf{\phi }(t)$ & A real-valued function $\phi $ on $\mathbb{R}%
^{+}\times \{e_{i}\}_{i=1}^{n}$ will be written as the vector $\mathbf{\phi }%
(t)\in \mathbb{R}^{n}$ as $\phi (t,x)=\mathbf{\phi }^{T}(t)x$ where $x\in
\{e_{i}\}_{i=1}^{n}$ \\ 
$\mathbf{\Phi }(t)$ & A real-valued function $\phi $ on $\mathbb{R}%
^{+}\times \{e_{i}\}_{i=1}^{r}$ $\times \{e_{i}\}_{i=1}^{n}$~is written as
the $r\times n$ real matrix $\mathbf{\Phi }(t)$ as $\phi (t,z,x)=z^{T}\Phi
(t)x$ where $z\in \{e_{i}\}_{i=1}^{r}$ and $x\in \{e_{i}\}_{i=1}^{n}$ \\ 
$A^{T}(z)K$ & Denotes the matrix whose $j^{\prime }th$ column is $%
A(e_{j})K^{T}e_{j}^{T}~\ \ $which can be more explicitly written as $%
\sum_{z}A^{T}(z)Kzz^{T}$ \\ 
$\left\vert M\right\vert $ & For a matrix $M$ represents the
element-by-element absolute value of a matrix \\ 
$M^{.2}$ & For a matrix $M$ represents the element-by-element squared \\ 
$\mathbf{e}_{r}$ & The $r-$vector $[1~1...1]^{T}$%
\end{tabular}

\section*{Acknowledgements}

Special thanks to Dr. Roger Brockett of Harvard School of Engineering and Applied Sciences who provided inspiration for the problems discussed in this paper, 
and to Dr. Andrew JK Phillips at Harvard Medical School for valuable feedback.

\bibliographystyle{acmtrans-ims}
\bibliography{manish}

\end{document}